\documentclass{article}

\usepackage{bbm}
\usepackage{amsmath}
\usepackage{amsthm}
\usepackage[margin=1.15in]{geometry}

\usepackage{graphicx}
\usepackage{hyperref}
\usepackage[capitalize]{cleveref}
\usepackage[numbers,sort,compress]{natbib}
\usepackage{xcolor}
\usepackage{enumerate}

\usepackage[small, margin=1cm]{caption}

\newcommand{\arrSet}{\mathbbm{A}}
\newcommand{\compSet}{\mathbbm{C}}
\newcommand{\indicator}[1]{\mathbbm{1}\{#1\}}
\newcommand{\Yarr}{Y_A}
\newcommand{\Ycomp}{Y_C}
\newcommand{\Ep}{\mathbbm{E}}
\newcommand{\abs}[1]{\lvert #1\rvert}

\newtheorem{lemma}{Lemma}[section]
\newtheorem{corollary}{Corollary}[section]
\newtheorem{theorem}{Theorem}[section]
\theoremstyle{definition}
\newtheorem{definition}{Definition}[section]
\theoremstyle{plain}

\makeatletter
\newtheorem*{rep@theorem}{\rep@title}
\newcommand{\newreptheorem}[2]{%
\newenvironment{rep#1}[1]{%
 \def\rep@title{#2 \ref{##1}}%
 \begin{rep@theorem}}%
 {\end{rep@theorem}}}
\makeatother

\newreptheorem{theorem}{Theorem}
\newreptheorem{lemma}{Lemma}
\newreptheorem{corollary}{Corollary}

\mathchardef\hy="2D

\title{Analysis of Markovian Arrivals and Service with Applications to Intermittent Overload}

\author{Isaac Grosof\footnote{Electrical and Computer Engineering Department, University of Illinois, Urbana-Champaign \& Department of Industrial Engineering and Management Science, Northwestern University, \url{izzy.grosof@northwestern.edu}}, Yige Hong\footnote{Computer Science Department, Carnegie Mellon University, \url{harchol@cs.cmu.edu}}, Mor Harchol-Balter\footnote{Computer Science Department, Carnegie Mellon University, \url{yigeh@andrew.cmu.edu}}}

\begin{document}

\maketitle

\begin{abstract}

In many important real-world queueing settings, arrival and service rates fluctuate over time.
We consider the MAMS system, where the arrival and service rates
each vary according to an arbitrary finite-state Markov chain,
allowing intermittent overload to be modeled.
This model has been extensively studied, and we derive results matching those found in the literature via
a somewhat novel framework.

We derive a characterization of mean queue length in the MAMS system,
with explicit bounds for all arrival and service chains at all loads,
using our new framework.
Our bounds are 
tight in heavy traffic.
We prove even stronger bounds for the important special case of
two-level arrivals with intermittent overload.

Our framework is based around the concepts of relative arrivals and relative completions,
which have previously been used in studying the MAMS system,
under different names.
These quantities allow us to tractably capture the transient correlational effect
of the arrival and service processes on the mean queue length.    
\end{abstract}

\section{Introduction}



Basic queueing theory assumes {\em independent and identically-distributed  (i.i.d.)} interarrival times and service durations.  In reality, the arrival rate might fluctuate between periods of high arrival rate and lower arrival rate. The same is true for the service rate.  
Because of these fluctuations, using a formula which is designed for a fixed arrival rate and service rate can lead to highly inaccurate results. 

There has been extensive work on characterizing queue length in systems where arrival rate and/or service rate fluctuate with time.
This technical report presents a somewhat new framework for analyzing such systems,
which we use to derive closed-form expressions for mean queue length in systems where the arrival rate and/or service rate fluctuates with time.  Specifically, we allow the arrival rate and service rate to each vary according to arbitrary finite-state Markov chains, and derive closed-form expressions for this setting.

Results equivalent to ours were first derived by \citet{yechiali_queuing_1971} in the two-level arrival setting,
derived by \citet{falin_heavy_1999} in the Markov-modulated arrivals setting,
and the results of Falin and Falin were extended to our Markov-modulated arrivals and service setting
by \citet{dimitrov_single_2011}.

The aforementioned results are stronger than ours, as they also characterize the entire moment generating function (MGF) of the queue length in heavy traffic.
We discuss these results further in \cref{sec:prior}.

Our framework is presented here in the hopes that it may be useful as a starting point to build future generalizations to more complex settings.

\subsection*{A specific motivating example: Intermittent overload}
As a motivating example, we study the problem of intermittent overload under a two-level arrival process.

Systems are often built so that the average load, $\rho$, is not too high.  This stems from the general understanding that mean queue length increases as a function related to $1/(1-\rho)$.
Standard operating principles for computer systems practitioners recommend keeping the {\em average} load under $0.5$ \citep{kearns_cpu_2003}.

However, a system with an average load of $0.5$ might not ever have an instantaneous load of $0.5$,
but rather it could be running at very low load most of the time, with brief periods of overload.
The example shown in \cref{fig:timevarying} is of a system which spends $\frac{5}{6}$ of the time at load $0.2$ and $\frac{1}{6}$ of the time at load $2.0$, resulting in an average load of $0.5$.
Because this system experiences intermittent overload, its mean queue length is {\em much larger} than that of a M/M/1 with load 0.5, especially when high-load and low-load periods are long. 

\begin{figure}
\centering
\includegraphics[width=.85\textwidth]{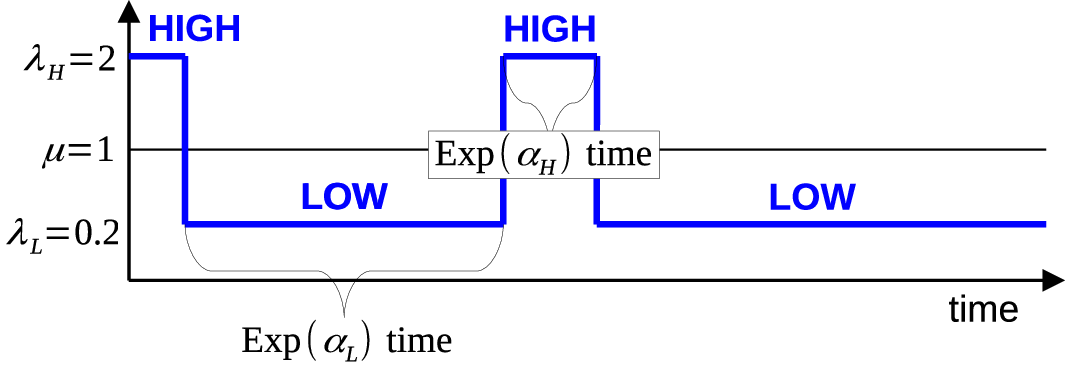}
\caption{A time-varying arrival process with intermittent overload. The system spends $5/6$ of its time on average with a low arrival rate of $\lambda_L=0.2$, and $1/6$ of its time with a high arrival rate of $\lambda_H=2$.  The average arrival rate is $\lambda = 0.5$, which is below the service rate of  $\mu=1$.}
\label{fig:timevarying}
\end{figure}

Intermittent overload, or fluctuating load more generally, is common in compute settings \citep{schroeder_web_2006,cho_mitigating_2023,hansson_overload_2001,farah_transient_2008,azim_efficient_2013,liu_breaking_2021} and also in medical setting settings \citep{curry_transient_2022,chan_queues_2016}. It is the reason that capacity provisioning is difficult, and it is the motivation behind the many papers written on dynamic capacity provisioning \citep{gandhi_autoscale_2012,qu_auto-scaling_2018,kaushik_cloud_2022}.

The simple two-level arrivals system has been extensively studied by \citet{yechiali_queuing_1971}, who derived a closed-form formula for mean queue length, as well as a tight result in heavy traffic.
We reanalyze this system with our framework, proving comparable results.

Further work on the two-level system also exists, which we discuss in \cref{sec:prior} \cite{gupta_fundamental_2006,vesilo_scaling_2022}.

\subsection*{Our results and approach}

{\bf Summary of results:}
This technical report proves closed-form bounds for mean queue length in the MAMS model,
using a somewhat new framework, matching existing bounds in the literature.
Our results are tight in heavy traffic and hold at all loads.   
When we specialize to the two-level system, we prove tight closed-form bounds in two additional settings:
in the limit as the system switches quickly between arrival rates ($\alpha_H, \alpha_L \to \infty$),
and in a setting with intermittent overload setting and slow switching between arrival rates ($\lambda_H > \mu; \alpha_H,\alpha_L \to 0$). 
Our main MAMS theorem is \cref{thm:e-q-exact-mams}, which we use to give closed-form bounds in \cref{cor:e-q-bounds-mams}, and our main two-level-arrival result is \cref{thm:e-q-exp-two}, which we use to give additional closed-form bounds in \cref{cor:two-level-e-q-bounds}.

{\bf Approach:}
Our framework is a somewhat novel extension of the drift method introduced by \citet{eryilmaz_asymptotically_2012}.
In a nutshell, the drift method involves exploiting the fact that the stationary drift of any given random variable related to system performance is zero.  The difficulty in using the drift method is choosing an appropriate test function which defines the random variable.  
In the past, many applications of the drift method have involved i.i.d. arrivals and service.
One exception is \citet{mou_switch_2020}, which required a complex extension to the drift method called the $m$-step drift method.  
Because the $m$-step drift method is much more complicated, the convergence rate in heavy traffic proven via this method is weaker than the existing result of \citet{falin_heavy_1999}.
Our new framework is finding a novel test function which allows for greater analytical tractability,
matching the existing literature.

Our new test function introduces a way of thinking about Markovian systems, which we call {\em relative arrivals} and {\em relative completions},
which is an independent rediscovery of the technique introduced by \citet{falin_heavy_1999}. The basic idea is as follows:  In an i.i.d. system, it is easy to quantify the expected number of arrivals by time $t$;  call this $\Ep[A(t)]$.  For a system with Markovian arrivals, such as the two-level-arrival system, the expected number of arrivals by time $t$ depends on which arrival state the system is in at time $0$:
If the system is in the high load state at time $0$, it will see more arrivals by time $t$ than if it is initially in the low load state.
If we let $A_H(t)$ (respectively, $A_L(t)$) denote the number of arrivals by time $t$ for a two-level system initialized in state $H$ (respectively, $L$), then we expect $\Ep[A_H(t)] > \Ep[A_L(t)]$.  We devise our novel test function by explicitly quantifying this difference.  Specifically, we define  the relative arrival quantities $\Delta(H) := \lim_{t \to \infty} \Ep[A_H(t)] - \Ep[A(t)]$ and 
$\Delta(L) := \lim_{t \to \infty} \Ep[A_L(t)] - \Ep[A(t)]$ (see \cref{def:rel-arr-comp-two}). In the general MAMS system, we also define relative completions similarly. These quantities are the heart of our novel test function.

{\bf Structure of presentation:}
To illustrate our method in the simplest setting possible, we start by presenting our results for the two-level arrival process, in \cref{sec:two-level-arr}.  This allows us to use a slightly simpler test function and also greatly simplify the notation and analysis, while being an interesting result in its own right. In the next several sections (\cref{sec:model,sec:results-mams,sec:mams-proofs}), we then solve the general MAMS system. In \cref{sec:two-level-bounds}, we return to the two-level system to tighten our bounds in this setting.

\subsection*{Outline of paper}

This technical report is organized as follows:
\begin{itemize}
    \item \cref{sec:prior}: We discuss prior work on the two-level arrivals system, and on queueing systems with Markovian arrivals and/or Markovian service.
    \item  \cref{sec:two-level-arr}: We introduce the two-level arrivals system, the drift method, and the relative arrivals function specialized to this setting. We prove our main result, \cref{thm:e-q-exp-two},
    for the two-level arrivals system.
    \item \cref{sec:model}: We introduce the MAMS model,
    and generalize the drift method and the relative arrivals and relative completions function to the MAMS setting.
    \item \cref{sec:results-mams}: We state our main result, \cref{thm:e-q-exact-mams}, for the general MAMS system.
    \item \cref{sec:mams-proofs}: We prove our main result for the general MAMS system.
    \item \cref{sec:two-level-bounds}: We prove additional bounds for the two-level arrivals system.
    \item \cref{sec:simulation}: We empirically demonstrate the accuracy and tightness of our bounds via simulation.
\end{itemize}

\section{Prior work}
\label{sec:prior}

We study the two-level arrival system. None of our results are novel -- all have previously appeared in the literature. We give a somewhat novel framework to derive these results.

\subsection{Results for two-level arrivals system}
\citet{yechiali_queuing_1971} study the same two-level system as the one studied in this technical report.
The paper it derives the explicit formula for the queue length and the steady-state distribution using a generating function technique
with some parameters determined by roots of cubic equations.
Their equation (33) exactly matches our main two-level result, \cref{thm:e-q-exp-two}.
The paper also studies the fast and slow switching settings that we study in \cref{sec:two-level-bounds},
in their Section II, cases C and D.

There has been additional study of the two-level system by \citet{gupta_fundamental_2006},
who prove that mean queue length decreases monotonically as switching rate increases, and proves stochastic ordering results, and by \citet{vesilo_scaling_2022}, who prove convexity and concavity results, additional stochastic ordering, and closing certain edge cases left open by \citet{gupta_fundamental_2006}.

\subsection{Bounds at all loads via Matrix Analytic Methods}

Matrix analytic methods  reduce the problem of deriving the stationary distribution of an MAMS system
to a series of linear algebraic calculations \citep{neuts_mm1_1978,ramaswami_ng1_1980,neuts_matrix_1981,latouche_introduction_1999}. Bounds can also be derived via spectral analysis of matrices related to the underlying system.

This line of work was initiated by \citet{neuts_mm1_1978},
which provided matrix equations to characterize a variety of performance metrics, including
busy period lengths, effective service times, virtual waiting times, time in system,
and queue length. For any given parameters of the system, appropriate matrices could be constructed.

An important follow-on work by \citet{regterschot_queue_1986}
used the Matrix Analytic approach to derive computationally tractable bounds, approximation, and analysis
of these performance quantities.

From these early matrix analytic works, an entire line of papers has followed, generalizing to more complex arrival and service processes, and improving the computational complexity of the algorithms.

A noteworthy recent result by \citet{ciucu_extensions_2018} gives a quite computationally efficient method to derive sharp bounds on the MGF of the queue length, for the MAMS system and a variety of other related systems, via Matrix Analytic methods.

\subsection{Heavy-traffic results for MAMS system}

The first heavy-traffic result in the MAMS system was by \citet{burman_mams_1986},
who studied a system with Markovian arrivals and general i.i.d. service time.
They proceed via a diffusion-scaling approach,
and prove a nearly-rigorous characterization of mean queue length in heavy traffic,
but do not prove a key interchange-of-limits lemma which would be necessary for a fully-rigorous proof.
Their heavy traffic results match our \cref{cor:e-q-bounds-mams}.

Next, \citet{falin_heavy_1999} study the same system as \cite{burman_mams_1986},
and rigorously characterize its queue length distribution,
including characterizing the limiting MGF in heavy traffic.
When specialized to the mean, their results match our main result, \cref{thm:e-q-exact-mams}.
They also introduce a concept equivalent to our relative arrivals quantity, which we describe in \cref{sec:model-rel-arr-two}. Their $a_m$ for an arrival state $m$ is equivalent to our $\Delta_A(m)$.
Their results fully subsume ours, and use a closely related approach to our approach.

\citet{dimitrov_single_2011} generalize the results of \citet{falin_heavy_1999}
to allow Markovian service rates, in addition to Markovian arrival rates and i.i.d. job sizes.
Their model differs slightly from ours, in that they use a single chain to modulate arrival and service rates,
but the key results match.

Our framework is based on the drift method of \citet{eryilmaz_asymptotically_2012}, incorporating the relative arrivals quantity.
While the relative arrivals quantity was an independent rediscovery of the approach by \citet{falin_heavy_1999},
the only prior work to use the drift method to analyze the MAMS system was  \citet{mou_switch_2020},
which analyzing a system with Markovian Arrivals and exponential service in the heavy traffic limit.
While the paper proves that the scaled mean queue length $E[Q(1-\rho)]$ converges in that limit,
the paper does not give a closed form expression for the limiting value,
leaving it in the form of an infinite summation which can be shown to be equivalent to the results of this technical report and of the prior work in this section.

\citet{mou_switch_2020}'s key analytical technique is an extension to the drift method \citep{eryilmaz_asymptotically_2012}
called the ``$m$-step drift method," which is considerably more complex than the original drift method.
As a result, their bounds are somewhat weaker, resulting in a slower heavy-traffic convergence rate.
By contrast, we extend the drift method in a different, simpler direction with our relative arrival and relative completion functions, while managing to prove tighter and more explicit results, matching the results of \citet{falin_heavy_1999} and \citet{dimitrov_single_2011} when restricted to mean queue length.

\section{Two-level arrivals}
\label{sec:two-level-arr}

In this section, we will introduce and analyze the two-level arrival system,
using a drift analysis based on our novel relative arrivals function.
Our main result in this section in \cref{thm:e-q-exp-two}, which characterizes the mean queue length in the two-level arrival system.
The techniques introduced in this section are then generalized to handle the entire MAMS model in \cref{sec:results-mams}.

We define the two-level arrivals model in \cref{sec:model-two}.
We introduce the drift method in \cref{sec:drift}.
We introduce our novel relative arrivals function in \cref{sec:model-rel-arr-two}. We sketch our proof in \cref{sec:test-func-two}, introducing the function whose drift we will analyze. We characterize its drift in \cref{sec:lemmas-two}, and use it to prove our main result \cref{sec:two-main-proof}.

\subsection{Two-level model}
\label{sec:model-two}

In the two-level arrivals model, jobs arrive to a single-server queue. Jobs wait in the queue in FCFS order. When a job reaches the head of the queue, it enters service. Job service times are exponentially distributed with rate $\mu$. All jobs are indistinguishable.

Jobs arrive according to a two-level process.
The two-level arrival process is depicted in \cref{fig:timevarying}. The arrival process consists of two arrival rates: A high rate $\lambda_H$ and a low rate $\lambda_L$. The system switches between the two rates according to the switching rates $\alpha_H$, the rate of switching from $H$ to $L$, and $\alpha_L$, the rate of switching from $L$ to $H$. A Markovian view on this arrival process is shown in \cref{fig:timevaryingarrivals}.
\begin{figure}
    \centering
\includegraphics[width=.5\textwidth]{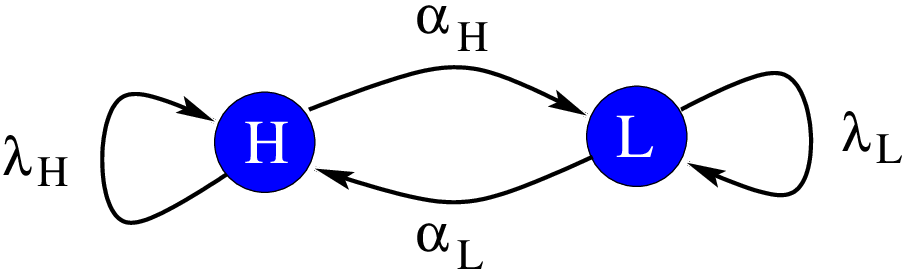}
\caption{The arrival chain of the two-level arrival process.}
\label{fig:timevaryingarrivals}
\end{figure}

Let $\lambda$ denote the long-term average arrival rate:
\begin{align}
    \label{eq:lambda-two}
    \lambda = \frac{\lambda_H/\alpha_H + \lambda_L/\alpha_L}{1/\alpha_H + 1/\alpha_L} = \frac{\lambda_H \alpha_L + \lambda_L \alpha_H}{\alpha_L + \alpha_H}.
\end{align}
Let $\rho := \lambda/\mu$ be the system load. 

We will write $y$ to denote a generic arrival state, $y \in \{H, L\}$.
We will write $(q, y)$ to denote a generic system state, where $q$ is the queue length and $y$ is the arrival state.
Let $\bar{y}$ denote the opposite arrival state from $y$: $\bar{H} = L$ and $\bar{L} = H$. 
Let $Y(t)$ denote the arrival state at time $t$, and let $Y$ be the corresponding stationary random variable.
We will also consider $Y^{arrival}$, the arrival-rate-weighted arrival state:
\begin{align}
    \label{eq:y-arrival-def}
    P(Y^{arrival} = y) := P(Y = y)\frac{\lambda_y}{\lambda}.
\end{align}

It is straightforward to compute the underlying distribution for $Y$.  We then get $Y^{arrival}$ from (\ref{eq:y-arrival-def}):

\begin{align}
    \label{eq:y-arrival-two}
    Y &= \begin{cases}
        H &\text{w.p. }  \frac{1/\alpha_H}{1/\alpha_H + 1/\alpha_L} = \frac{\alpha_L}{\alpha_H + \alpha_L} \\
        L &\text{w.p. }  \frac{1/\alpha_L}{1/\alpha_H + 1/\alpha_L} = \frac{\alpha_H}{\alpha_H + \alpha_L}
    \end{cases}, \qquad 
    Y^{arrival} = \begin{cases}
        H &\text{w.p. }  \frac{\alpha_L \lambda_H}{\alpha_H \lambda_L + \alpha_L \lambda_H} \\
        L &\text{w.p. }  \frac{\alpha_H \lambda_L}{\alpha_H \lambda_L + \alpha_L \lambda_H}
    \end{cases}.
\end{align}

\subsection{Drift method}
\label{sec:drift}

A key idea for analyzing queueing systems is that for any stationary random variable, its rate of increase must equal its rate of decrease. For instance, consider the work $W$ in an M/G/1. Using the fact that the rate of increase due to arriving jobs must equal the rate of decrease due to service, one can prove that $\rho$, the fraction of time the server is occupied, must equal $\lambda \Ep[S]$, the product of the arrival rate and the mean job size.
This balance holds for any stationary random variable. In the M/G/1, by analyzing the rate of increase and decrease of $W^2$, one can characterize the mean work in the system.

To formalize this concept and apply it to our system,
we make use of the \emph{drift} of a random variable.
In a given state, the random variable's drift is its instantaneous rate of change, in expectation over the randomness of the system. The key fact about drift is that in stationarity, the expected drift of any random variable is $0$, where the expectation is taken both over the state and the future randomness of the system.

Formally, this is captured by the \emph{instantaneous generator} $G$,
which is the stochastic equivalent of the derivative operator.
Rather than operating on a random variable directly, it is simpler to think of $G$ as operating on a \emph{test function}, a function $f$ that maps system states $(q, y)$ to real values.
The instantaneous generator takes $f$ and outputs $G \circ f$, the drift of $f$.

Let $G$ specifically denote the generator operator for the two-level arrival system, which is defined as follows: 
For any test function $f(q, y)$, 
\begin{align*}
    G \circ f(q, y) := \lim_{t \to 0} \frac{1}{t} \Ep[&f(Q(t), Y(t)) - f(q, y) \mid Q(0) = q, Y(0)=y].
\end{align*}

Now, we can state the key fact about the drift:
\begin{lemma}\label{lem:drift-lemma-two}
    Let $f$ be a real-valued function of the MAMS system state where $\Ep[f(Q, \Yarr, \Ycomp)] < \infty.$
    Then 
    \begin{equation}\label{eq:generator-steady-zero}
        \Ep [G \circ f(Q, Y)] = 0,
    \end{equation}
    where the expectation is taken over the stationary random variables $Q, Y$.
\end{lemma}
    This result is a special case of the more general result for the MAMS system, \cref{lem:drift-lemma-mams}, which we prove in \cref{app:lemma-generator-drift}.
The \emph{drift method}
relies on finding a useful test function $f$, such that the fact that the expected drift is zero allows us to analyze the performance of the system.

\subsection{Relative arrivals}
\label{sec:model-rel-arr-two}

Next, we introduce a concept named the \emph{relative arrival} function, $\Delta(i)$, which will allow us to define our needed test function for the drift method. This function was used in \cite{falin_heavy_1999} to analyze the M/G/1 with Markov-modulated arrivals, where it was referred to as $a_i$.

Let $A_y(t)$ denote the number of arrivals by time $t$ given that the arrival process starts in state $y\in\{H,L\}$.
\begin{definition}
    \label{def:rel-arr-comp-two}
    The relative arrivals function $\Delta(y)$ is defined as follows:
\begin{align*}
    \Delta(H) &:= \lim_{t\to\infty} \Ep[A_H(t)] - \lambda t, \qquad
    \Delta(L) := \lim_{t\to\infty} \Ep[A_L(t)] - \lambda t.
\end{align*}
\end{definition}

\begin{figure}
    \centering
    \includegraphics[width=0.7\textwidth]{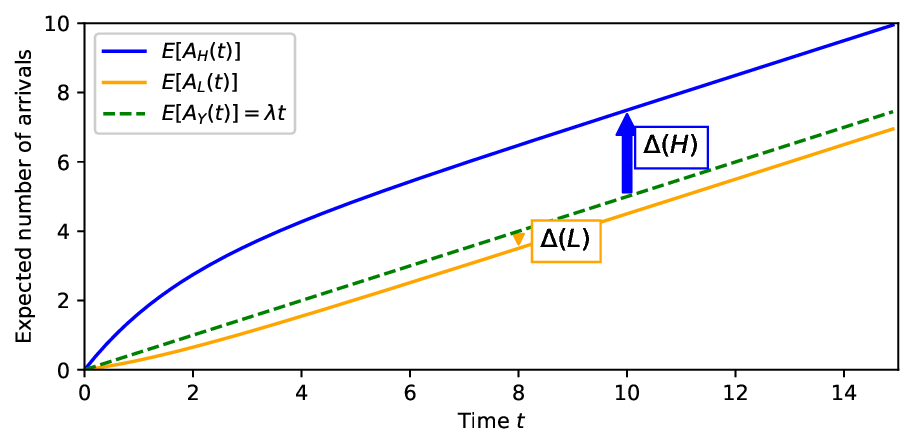}
    \caption{Expected arrivals by time $t$, starting in arrival states $H$ and $L$, and in the steady-state $Y$. Relative arrivals $\Delta(H)$ and $\Delta(L)$ are the limiting differences between these functions. Setting: $\lambda_H = 2, \lambda_L=0.2, \alpha_H = 0.5,  \alpha_L=0.1, \mu=1, \lambda = 0.5$.}
    \label{fig:rel-arr-two}
\end{figure}

As we illustrate in \cref{fig:rel-arr-two}, the relative arrivals function, $\Delta(y)$, represents the difference between the expected number of arrivals given that we're starting in state $y$ and the expected total number of arrivals. As time goes to infinity, this difference converges to a finite value.

$\Delta(H)$ and $\Delta(L)$ are related as follows:
\begin{align}
    \label{eq:poisson-two}
    \Delta(H) &= \frac{\lambda_H - \lambda}{\alpha_H} + \Delta(L), \qquad
    \Delta(L) = \frac{\lambda_L - \lambda}{\alpha_L} + \Delta(H).
\end{align}

To see why \eqref{eq:poisson-two} holds, note that $\lambda_H - \lambda$ is the difference in arrival rates between $H$ and the long-term rate, and $\frac{1}{\alpha_H}$ is the time for which that difference accrues. After this time, the system transitions to state $L$.

To characterize the exact values of $\Delta(H)$ and $\Delta(L)$, we can use the fact that the time-average value of $\Delta$ is zero: $\Ep[\Delta(Y)] = 0$, which we prove for the more general MAMS system in \cref{lem:mean-delta}.
As a result:
\begin{align}
    \label{eq:delta-calc-two}
    \Delta(H) &= \alpha_H \frac{\lambda_H - \lambda_L}{(\alpha_H + \alpha_L)^2},  \qquad 
    \Delta(L) = \alpha_L \frac{\lambda_L - \lambda_H}{(\alpha_H + \alpha_L)^2}.
\end{align}

The drift of the relative arrivals function is nicely behaved: $G \circ \Delta(y) = \lambda - \lambda_y$.
This follows from \cref{def:rel-arr-comp-two}:
The $\Ep[A_y(t)]$ term contributes $\lambda_y$ to the drift, while the $\lambda t$ term contributes $\lambda$. 
This property is generalized to the MAMS system and proven in \cref{cor:effect-arr-comp-on-rel}.

\subsection{Developing our test function}
\label{sec:test-func-two}

We come up with a test function $f_{\Delta}(q,y)$ for the drift method that combines both the queue length and the relative arrivals. 
Although the relative arrival function is not new, our test function is novel. 

\begin{definition}\label{def:f-delta-two}
    $
        f_{\Delta}(q, y) := (q + \Delta(y))^2 - \Delta(y)^2 = q(q + 2\Delta(y)). 
    $
\end{definition}

The intuition for the use of the test function $f_\Delta(q, y)$ is that on its own, the queue length $q$ changes at a state-dependent rate, which is hard to analyze. Specifically, jobs arrive at the rate $\lambda_y$ and complete at the rate $\mu$, so the queue length $q$ has drift $\lambda_y - \mu$, assuming $q > 0$. 

We want to ``smooth out'' the instantaneous rate of change of $q$. Specifically, we want a function which:
\begin{itemize}
    \item has a constant \emph{drift}, in all system states where $q>0$, and
    \item differs from $q$ by a bounded amount.
\end{itemize}
Recall from \cref{sec:model-rel-arr-two} that the drift of $\Delta(y)$ is $\lambda - \lambda_y$. 
Thus, the function $q + \Delta(y)$ has a constant drift of $\lambda - \mu$ in all system states where $q > 0$.

Thus, $q+\Delta(y)$ is a smoothed-out proxy for the queue length $q$ with constant drift, making it tractable for analysis.
We set $f_{\Delta}(q, y)$ to be a smoothed-out proxy for the \emph{square} of the queue length,
as applying the generator $G$ gives information about the expected derivative
of the test function, and hence information about $\Ep[Q]$.
Specifically, as we show in \cref{lem:G-f-Delta-two-level}, the drift of $f_{\Delta}(q, y)$ separates cleanly: It has a $(\lambda-\mu)q$ term, and terms that do not depend on $q$.
Then, in \cref{thm:e-q-exp-two}, we apply the fact that the drift is zero in stationarity (\cref{lem:drift-lemma-two}), allowing us to characterize $\Ep[Q]$.

\subsection{Lemmas for the two-level system}
\label{sec:lemmas-two}

First, let's calculate the drift of a generic test function in the two-level system.

\begin{lemma}\label{lem:G-f-two-level}
For any test function $f$ in the two-level system, its instantaneous drift in a particular system state $(q,y)$ can be expressed as 
\begin{align*}
    &G\circ f(q, y) 
    = \lambda_y (f(q + 1, y) - f(q, y))
    + \alpha_y (f(q, \bar{y}) - f(q, y))
     + \mu( f(q - 1 + u, y) - f(q, y)),
\end{align*}
where $u := (1-q)^+ = \indicator{q=0}$ denotes unused service, and $\bar{y}$ is the opposite arrival state from $y$.
\end{lemma}
\begin{proof}
    In a given state of the two-level system $(q, y)$, there are three events that can occur:
    arrival, change of the arrival state $y$, and completion, if $q>0$.
    These three events each contribute a corresponding term to the drift of the generic test function $f$.
\end{proof}

Now, we can calculate the drift for the special test function $f_{\Delta}(q, y) = q(q + 2\Delta(y))$.
\begin{lemma}\label{lem:G-f-Delta-two-level}
For any system state $(q, y)$ of the two-level system, the instantaneous drift is
\begin{align}
     &G\circ f_{\Delta}(q, y)  = 2(\lambda - \mu) q  + 2(\lambda_y - \mu (1-u))\Delta(y)  + \lambda_y + \mu (1-u).
\end{align}
\end{lemma}

\begin{proof}
    By \Cref{lem:G-f-two-level} and the definition of $f_{\Delta}(q, y)$, 
    \begin{align}
        G\circ f_{\Delta}(q, y)
        =\, &
        \label{eq:two-level-G-f-Delta-lambda-term}
        \lambda_y \big((q + 1)(q+1 + 2\Delta(y))  -  q(q + 2\Delta(y))\big) \\
        \label{eq:two-level-G-f-Delta-alpha-term} 
        &+ \alpha_y \big(q(q + 2\Delta(\bar{y}))  - q(q + 2\Delta(y))\big)\\
        \label{eq:two-level-G-f-Delta-mu-term} 
        &+ \mu\big((q-1+u)(q-1+u + 2\Delta(y))  - q(q + 2\Delta(y))\big). 
    \end{align}
    Let us simplify each term in the expression. Each term will simplify to an affine function of $q$:
    A linear term in $q$, plus an additive offset. The slope of some of these affine functions will depend on $y$,
    but we will show that for the overall drift, the slope does not depend on $y$,
    due to our choice of test function.
    
    For the arrivals term, \eqref{eq:two-level-G-f-Delta-lambda-term}, 
    \begin{align*}
        &  \lambda_y \big((q + 1)(q+1 + 2\Delta(y))  -  q(q + 2\Delta(y))\big) 
        =\, 2 \lambda_y q + \lambda_y \big(2\Delta(y) + 1\big). 
    \end{align*}
    
    For the switching term, \eqref{eq:two-level-G-f-Delta-alpha-term}, 
    \begin{align*}
         & \alpha_y \big(q(q + 2\Delta(\bar{y}))  - q(q + 2\Delta(y))\big) = 2\alpha_y q(\Delta_{A}(\bar{y}) - \Delta(y)). 
    \end{align*}
    
    For the completions term, \eqref{eq:two-level-G-f-Delta-mu-term},
    \begin{align*}
         \mu\big((q-1+u)(q-1+u + 2\Delta(y))  - q(q + 2\Delta(y))\big) = -2\mu q + \mu (1-u) \big(-2\Delta(y) + 1\big),
    \end{align*}
    where we make use of the facts that $qu = 0$ and that $(1-u)^2 = (1-u)$, noting that $(1-u) \in \{0,1\}$. 
    
    Combining the above calculations, we find that
    \begin{align}
        G\circ f_{\Delta}(q, y)
        \label{eq:linear-q-two}
        &= 2(\lambda_y + \alpha_y (\Delta(\bar{y}) - \Delta(y)) - \mu) q \\
        &+ 2(\lambda_y - \mu (1-u))\Delta(y)  + \lambda_y + \mu (1-u).
        \nonumber
    \end{align}

    Now, we want to simplify the $\lambda_y + \alpha_y (\Delta(\bar{y}) - \Delta(y))$ coefficient in \eqref{eq:linear-q-two}.
    To do so, let us use a relationship demonstrated in \cref{sec:model-rel-arr-two},
    namely \eqref{eq:poisson-two}:
    \begin{align*}
        \Delta(y) &= \frac{\lambda_y - \lambda}{\alpha_y} + \Delta(\bar{y}) \implies
        \lambda = \lambda_y + \alpha_y (\Delta(\bar{y}) - \Delta(y)). 
    \end{align*}
    Thus, the linear term in \eqref{eq:linear-q-two} simplifies to a term which does not depend on $y$, as desired:
    \begin{align*}
        &G\circ f_{\Delta}(q, y)  = 2(\lambda - \mu) q  + 2(\lambda_y - \mu (1-u))\Delta(y)  + \lambda_y + \mu (1-u).
        \qedhere
    \end{align*}
\end{proof}

\subsection{Main result for two-level system}
\label{sec:two-main-proof}

\begin{theorem}
    \label{thm:e-q-exp-two}
    In the two-level arrivals system  with exponential service rate,
    the mean queue length is:
    \begin{align}
        \label{eq:two-main-statement}
        \Ep[Q] &= \frac{\rho(\Ep[\Delta(Y^{arrival})]+1)}{1-\rho}
        + \Ep[\Delta(Y)\mid Q=0], \text{ where} \\
        \nonumber
        \Ep[\Delta(Y^{arrival})] &= \frac{(\lambda_H - \lambda_L)^2 \alpha_L \alpha_H}{\lambda (\alpha_H+\alpha_L)^3}, \\
        \nonumber
         \Ep[\Delta(Y)\mid Q=0] &= \frac{\lambda_H - \lambda_L}{\alpha_H + \alpha_L} \left( P(Y=H\mid Q=0) - \frac{\alpha_L}{\alpha_H + \alpha_L}\right). 
    \end{align}
\end{theorem}

\begin{proof}
    We start with the result of \cref{lem:G-f-Delta-two-level}:
    \begin{align}
        \label{eq:drift-two-level}
        G\circ f_{\Delta}(q, y)  = 2(\lambda - \mu) q  + 2(\lambda_y - \mu (1-u))\Delta(y)  + \lambda_y + \mu (1-u). 
    \end{align}

    Recall \cref{lem:drift-lemma-two}, the key drift-method lemma,
    which states that for (essentially) any test function, $f(q, y)$, we have $\Ep[G\circ f(Q, Y)] = 0$, where $(Q, Y)$ is the stationary random variable for the system state.
    In particular, $\Ep[G\circ f_{\Delta}(Q, Y)]=0$. We will use this fact to characterize $\Ep[Q]$.
    More specifically, because $f_{\Delta}(q, y)$ grows polynomially in $q$, we show it satisfies \Cref{lem:drift-lemma-two} in \Cref{lem:polynomial-finite-expectation}.
    
    Taking the expectation of \eqref{eq:drift-two-level}, we find that
    \begin{align}
    \nonumber
        0 &= \Ep[G\circ f_{\Delta}(Q, Y)] \\
    \nonumber
        &= 2(\lambda - \mu) \Ep[Q] + 2\Ep_{y \sim Y}[\lambda_y \Delta(y)]
        - 2\mu \Ep[\Delta(y)]
        + 2\mu \Ep[\Delta(Y)\indicator{Q=0}]+ \lambda + \mu P(Q>0) \\
    \label{eq:drift-two-exp}
        &= 2(\lambda - \mu) \Ep[Q] + 2\Ep_{y \sim Y}[\lambda_y \Delta(y)]
        + 2\mu \Ep[\Delta(Y)\indicator{Q=0}]+ \lambda + \mu P(Q>0).
    \end{align}
    In taking this expectation, we use the facts that $\lambda := \Ep[\lambda_Y]$,
    that $\Ep[\Delta(Y)]=0$,
    and that $u := \indicator{q=0}$.
    Note that the $(\lambda-\mu)q$ term in \eqref{eq:drift-two-level}
    allows us to get a $\Ep[Q]$ term in \eqref{eq:drift-two-exp}.

    To simplify \eqref{eq:drift-two-exp}, note that because the system has an exponential service rate, $P(Q>0) = \rho = \frac{\lambda}{\mu}$,
    by \cref{lem:unused}.
    Thus:
    \begin{align}
        0 =\,& (\lambda - \mu) \Ep[Q] + \Ep_{y \sim Y}[\lambda_y \Delta(y)] + (\mu - \lambda) \Ep[\Delta(Y)\mid Q=0]+ \lambda.
        \label{eq:drift-two-simp}
    \end{align}

    Recall that we defined $Y^{arrival}$ to be the arrival-weighted state of the arrival process.
    In particular, from \eqref{eq:y-arrival-def}, $P(Y^{arrival} = y) := P(Y = y)\frac{\lambda_y}{\lambda}$. 
    As a result, $\Ep_{y \sim Y}[\lambda_y \Delta(y)] = \lambda \Ep[\Delta(Y^{arrival})]$,
    and \eqref{eq:drift-two-simp} becomes:
    \begin{align*}
        0 &= (\lambda - \mu) \Ep[Q] + \lambda \Ep[\Delta(Y^{arrival})] + (\mu - \lambda) \Ep[\Delta(Y)\mid Q=0]+ \lambda \\
        \Ep[Q] &= \frac{\rho(\Ep[\Delta(Y^{arrival})] + 1)}{1-\rho} + \Ep[\Delta(Y)\mid Q=0].
    \end{align*}

    Using the values of $\Delta$ from \eqref{eq:delta-calc-two} and the distribution for $Y^{arrival}$ from \eqref{eq:y-arrival-two}, we can compute $\Ep[\Delta(Y^{arrival})]$ and compute $\Ep[\Delta(Y)\mid Q=0]$ up to the value of $P(Y = H \mid Q =0)$, which we bound in \cref{sec:two-level-bounds}.
\end{proof}

Note that plugging in the bound $P(Y=H\mid Q=0) > 0$ immediately gives a strong lower bound on $\Ep[Q]$.
In \cref{sec:two-level-bounds}, we prove tight upper bounds on $P(Y=H\mid Q=0)$ in both the quick switching limit ($\alpha_H, \alpha_L \to \infty$) and in the slow switching limit ($\alpha_H, \alpha_L \to 0$), giving tight results in both regimes.

Even without any bounds on $P(Y=H\mid Q=0)$, \cref{thm:e-q-exp-two} already tightly characterizes mean response time in the $\rho \to 1$ limit (the heavy traffic limit), holding $\alpha_H$ and $\alpha_L$ constant:

\begin{corollary}\label{cor:heavy-traffic-e-q-two-level}
    In the two-level arrivals system, the scaled mean queue length in the heavy traffic limit converges to $\lim_{\rho\to 1} (1-\rho) \Ep[Q] = \Ep[\Delta(Y^{arrival})] + 1.$
\end{corollary}
This holds because $\Ep[\Delta(Y)\mid Q=0]$ is bounded by $\Delta(L)$ and $\Delta(H)$, which are bounded in the $\rho \to 1$ limit.

In the more general setting of Markovian arrivals and Markovian service (MAMS), we prove \cref{thm:e-q-exact-mams}, a generalization of our two-level result \cref{thm:e-q-exp-two}, which can handle any number of arrival rates, as well as Markov modulated service rates. This results in a similar heavy traffic result, \cref{cor:heavy-traffic-e-q-mams}.

\section{MAMS model and definitions}
\label{sec:model}

We now turn to the general setting of Markovian arrivals and Markovian service (MAMS), where the arrival process and the completions process are both Markov chains.

Like the two-level arrival system introduced in \cref{sec:model-two}, the MAMS system consists of a single-server queue. Jobs wait in the queue in FCFS order. When a job reaches the head of the queue, it enters service. After some time in service, the job completes. All jobs are indistinguishable.
Jobs arrive according to an arrival process which we define in \cref{sec:arrival-def-mams},
and complete according to a completion process which we define in \cref{sec:completion-def-mams}.
After introducing the MAMS model,
\cref{sec:drift-mams,sec:model-rel-arr-mams,sec:test-func} generalize our discussion of the drift process,
the notion of relative arrivals (and completions),
and the test function from what we saw in Section~\ref{sec:two-level-arr}.  The main results for the MAMS model are stated  in Section~\ref{sec:results-mams}.  

Our goal is to bound the mean queue length in the MAMS system, with bounds that are tight in heavy traffic. We seek to generalize \cref{thm:e-q-exp-two}, our main result for the two-level arrivals system.

\subsection{Arrivals}
\label{sec:arrival-def-mams}
We start by generalizing the arrival process to be a general finite-state continuous-time Markov chain.
In \cref{sec:model-two}, the arrival process had only two states, high arrival intensity and low arrival intensity. There were two types of updates in that arrival process: Jobs arriving, which happened with rate $\lambda_H$ and $\lambda_L$, and state changes, which happened with rate $\alpha_H$ and $\alpha_L$.

In the more general MAMS setting, we must track both the initial and final state of a state change,
and we also allow for a job to arrive simultaneously with the state change.

We thus use the following more general notation:
Let the state space of the arrival chain be some finite set $\arrSet$.
For a given pair of states $i \neq j \in \arrSet$,
the system may transition from $i$ to $j$ with an arrival, or from $i$ to $j$ without an arrival.
We will denote the rate of the former transition as $r_{i,j,1}$, and of the later transition as $r_{i,j,0}$.
In general, we will write the transition rate as $r_{i,j,a}$, where the $a$ parameter indicates whether an arrival occurs during a transition.
As a special case, self-transitions may only occur with an arrival: $r_{i,i,1}$ may be nonzero, $r_{i,i,0}$ is always zero.

For example, in the two-level system of \cref{sec:two-level-arr},
$\arrSet = \{H,L\}$, and the nonzero transition rates are:
\begin{align*}
          r_{H,H,1} = \lambda_H, \quad r_{L,L,1} = \lambda_L,
    \quad r_{H,L,0} = \alpha_H, \quad r_{L,H,0} = \alpha_L. 
\end{align*}

We assume that the arrival chain is irreducible.
Let $\lambda$ be the long-term arrival rate, as in \cref{sec:model-two}.
Let $\Yarr(t)$ be a random variable denoting the state of the arrival chain at time $t$.
Let $\Yarr$ be the corresponding stationary random variable.
Let $\Yarr^{arrival}$ be the arrival-rate-weighted arrival state, defined as follows:
\begin{align*}
    P(\Yarr^{arrival} = j) := \frac{\sum_{i\in\arrSet} P(\Yarr = i) r_{i,j,1}}{\lambda}. 
\end{align*}
Note that $\Yarr(t), \Yarr,$ and $\Yarr^{arrival}$ correspond to $Y(t), Y,$ and $Y^{arrival}$
from \cref{sec:model-two}. The subscript $_A$ is added to differentiate the arrival chain from the completion chain, which will use the subscript $_C$.

We will use $\cdot$ in the rate subscripts when aggregating multiple related rates.
For instance, $r_{i,\cdot,\cdot} := \sum_{j,a} r_{i,j,a}$.
As a special case, we will define $\lambda_i = r_{i,\cdot,1}$ to be the total instantaneous arrival rate in state $i$.

\subsection{Completions}
\label{sec:completion-def-mams}
In the two-level arrivals model in \cref{sec:two-level-arr}, completions always occurred at rate $Exp(\mu)$. In the general MAMS model, we allow the completion chain to be a general finite-state continuous-time Markov chain.

Paralleling the arrival chain defined in \cref{sec:arrival-def-mams},
we let the state space of the completion chain be some finite set $\compSet$.
For each $i, j \in \compSet$, the system may transition from $i$ to $j$ with a completion or without a completion, at rates $s_{i,j,0}$ and $s_{i,j,1}$ respectively.
As before, $s_{i,i,1}$ may be nonzero, but $s_{i,i,0} = 0$. 
We write these rates as $s_{i,j,c}$, where the $c$ parameter denotes whether a completion occurs during a transition.
Again, we assume the completion chain is irreducible.

One important note: A completion transition may occur when there are no jobs present in the system.
If this happens, the completion process's state change still occurs, and the queue remains empty.
We will refer to such a completion as an ``unused completion."

Let $\mu$ denote the long-term rate of completions, whether used or unused.

As an example, in the two-level system there was only a single completion state, which we may arbitrarily call $X$. Then $\compSet=\{X\}$.
The sole completion transition rate was $s_{X,X,1}=\mu$.

We define $\Ycomp(t), \Ycomp$, and $\Ycomp^{comp}$ to correspond to $\Yarr(t), \Yarr,$ and $\Yarr^{arrival}$ from \cref{sec:arrival-def-mams}. In particular,
\begin{align*}
    P(\Ycomp^{comp} = j) := \frac{\sum_{i\in\compSet} P(\Ycomp = i) s_{i,j,1}}{\mu}. 
\end{align*}

We define the load $\rho$ to be the ratio $\lambda/\mu$. Our bounds hold for all loads $\rho$, and are tight in the heavy-traffic limit, as $\rho \to 1$.

In our main result for the two-level arrival system, \cref{thm:e-q-exp-two}, we make use of an expectation of the form $\Ep[\cdot \mid Q = 0]$,
conditioning on the system state when the queue is empty.
In the MAMS setting, we condition the system state when when unused completions occur.
We will write this expectation as $\Ep_U[\cdot]$.
Note that in the two-level system, because of the exponential service process and the PASTA property \citep{wolff_pasta_1982},
these two expectations coincide: $\Ep_U[\cdot]=\Ep[\cdot \mid Q = 0]$.
We define the expectation $\Ep_U[\cdot]$ in more detail in \cref{app:unused}.

As with the arrival process, we use $\cdot$ to aggregate related rates, and define $\mu_i = s_{i,\cdot,1}$ to correspond to $\lambda_i$.

\subsection{Drift method}
\label{sec:drift-mams}

In the general MAMS model, we continue using the drift method which we discussed in Section~\ref{sec:drift}, using the fact that the rate of change of any random variable is zero in steady state. 
Formally, the rate of change of a random variable is captured by the drift of a test function. 
Here a test function $f$ is any real-valued function of the state of the MAMS system, $(q, i_A, i_C)$, where $q$ denotes the total number of jobs, and $i_A$ and $i_C$ denote the states of the arrival and service processes, respectively.   
The drift of the test function $ G \circ f(q, i_A, i_C) $ is defined as 
\begin{align*}
    G \circ f(q, i_A, i_C) := \lim_{t \to 0} \frac{1}{t} \Ep[&f(Q(t), Y_A(t), Y_C(t)) - f(q, i_A, i_C)
    \mid Q(0) = q, Y_A(0)=i_A, Y_C(0)=i_C].
\end{align*}
The operator $G$ is called the \emph{instantaneous generator}, which takes a test function $f$ as input and gives the drift of the test function $G\circ f$ as output. The instantaneous generator can be seen as the stochastic equivalent of the derivative operator. 

The lemma below shows that the expected value of the $G\circ f$ in steady state is zero, generalizing \Cref{lem:drift-lemma-two} to the MAMS system. We prove \Cref{lem:drift-lemma-mams} in \Cref{app:lemma-generator-drift}. 
\begin{lemma}\label{lem:drift-lemma-mams}
    Let $f$ be a test function for which $\Ep[f(Q, \Yarr, \Ycomp)] < \infty.$
    Then 
    \begin{equation}\label{eq:generator-steady-zero-mams}
        \Ep [G \circ f(Q, \Yarr, \Ycomp)] = 0.
    \end{equation}
\end{lemma}

All test functions $f$ that we will use with \eqref{eq:generator-steady-zero-mams} will grow at a polynomial rate in $q$, which we show satisfy \Cref{lem:drift-lemma-mams} in \Cref{lem:polynomial-finite-expectation}. 

Similar to \Cref{lem:G-f-two-level}, \Cref{lem:g-f-formula-mams} below provides the expression for the drift of a generic function in the MAMS system. \Cref{lem:g-f-formula-mams} is standard in the drift method literature, and we have specialized it to this system. 
We prove \Cref{lem:g-f-formula-mams} in \Cref{app:lemma-generator-drift}. 

\begin{lemma}\label{lem:g-f-formula-mams}
    For any real-valued function $f$ of the state of the MAMS system, 
    \begin{align*}
         G \circ f(q, i_A, i_C) &= \sum_{j_A\in\arrSet, a\in\{0,1\}} r_{i_A, j_A, a} \left(f(q+a, j_A, i_C) -  f(q, i_A, i_C)\right) \\
         & + \sum_{j_C\in\compSet, c \in\{0,1\}} s_{i_C, j_C, c } \left(f(q-c+u, i_A, j_C) - f(q, i_A, i_C)\right),
    \end{align*}
    where $u=(c-q)^+$ denotes the unused service. 
\end{lemma}

\subsection{Relative arrivals and completions}
\label{sec:model-rel-arr-mams}

Next, we introduce a concept named the \emph{relative arrival} function, $\Delta(i)$, which will allow us to define our needed test function for the drift method. 

In \cref{sec:model-rel-arr-two}, we introduced the concept of the relative arrivals function,
which was used in \cite{falin_heavy_1999} to analyze the M/G/1 with Markov-modulated arrivals, where it was referred to as $a_i$.
In the MAMS setting, we denote this relative arrivals function as $\Delta_A(i)$,
and we also introduce the relative completions function $\Delta_C(i)$, defined equivalently.

Let $A_i(t)$ denote the number of arrivals by time $t$ given that the arrival chain starts in state $i \in \arrSet$ at time $t=0$.
Let $C_i(t)$ denote the  number of completions by time $t$ given that the completion chain starts in state $i \in \compSet$ at time $t=0$.

\begin{definition}
    \label{def:rel-arr-comp}
    Define the relative arrivals and relative completions functions $\Delta_A(i), \Delta_C(i)$ as follows:
\begin{align*}
    \Delta_A(i) &:= \lim_{t\to\infty} \Ep[A_i(t)] - \lambda t,  \qquad  
    \Delta_C(i) := \lim_{t\to\infty} \Ep[C_i(t)] - \mu t.
\end{align*}
\end{definition}

We verify that these limits always converge to a finite value in \cref{app:rel-arr-comp}. 
These functions can be seen as the relative value function of the Markov reward process where a reward of 1 is received whenever an arrival or completion occurs, respectively.
We can therefore calculate their values explicitly by solving the corresponding Poisson equation,
given in \cref{lem:poisson-delta}, along with \cref{lem:mean-delta}, which states that these quantities have zero expected value in stationarity:

\begin{lemma}
    \label{lem:mean-delta}
    $\Ep[\Delta_A(\Yarr)] = 0$, $\Ep[\Delta_C(\Ycomp)] = 0$.
\end{lemma}

\begin{proof}
    Note that $[\Yarr(t) \mid \Yarr(0) = \Yarr] \sim \Yarr$.
    As a result, $\Ep[A_{\Yarr}(t)] = \lambda t$ for any $t$. Then by definition, $\Ep[\Delta_A(\Yarr)] = \lim_{t\to\infty} (\Ep[A_{\Yarr}(t)] - \lambda t) = 0$. Similarly, $\Ep[\Delta_C(\Ycomp)] = 0$. 
\end{proof}

Another important property of $\Delta_A(i)$ and $\Delta_C(i)$
are their drifts:
\begin{repcorollary}{cor:effect-arr-comp-on-rel}$G\circ \Delta_A(i_A) = \lambda - \lambda_{i_A}, G\circ \Delta_C(i_C) =  \mu - \mu_{i_C}.$
\end{repcorollary}

\cref{cor:effect-arr-comp-on-rel} follows from the Poisson equation, \cref{lem:poisson-delta}, and it is proven in \cref{app:rel-arr-comp}.

\subsection{Test function}
\label{sec:test-func}

In order to characterize $\Ep[Q]$, we introduce the test function $f_{\Delta_A,\Delta_C}(q, i_A, i_C)$, which generalizes the two-level arrival test function defined in \Cref{def:f-delta-two} to the MAMS system. 
\begin{definition}
    $f_{\Delta_A,\Delta_C}(q, i_A, i_C) := q(q + 2 \Delta_A(i_A) - 2\Delta_C(i_C)).$
\end{definition}

The intuition for this test function again comes from finding a smooth proxy for the queue length $q$.
In the two-level system, we used $q+\Delta(y)$ as a proxy for $q$, to capture the influence of the state-dependent arrival rate in the two-level system,
In the MAMS system, both the arrival rates and completion rates are state-dependent, so we use  $q+\Delta_A(i_A)-\Delta_C(i_C)$ as a smooth proxy for $q$. 

To characterize mean queue length, we need a proxy for the square of the queue length. In the two-level system, we used $q(q+2\Delta(y))$ as our proxy for the square of queue length.
In the MAMS system, $f_{\Delta_A, \Delta_C}(q, i_A, i_C)$ fulfils the same role.

\section{Results}
\label{sec:results-mams}

We now present the main of the paper: a closed-form formula for the  mean queue length in the MAMS system.   This will be proven in Section~\ref{sec:mams-proofs}.

\begin{theorem}[Mean queue length]
\label{thm:e-q-exact-mams}
    In the MAMS system, the mean queue length is given by:
    \begin{align*}
        &\Ep[Q] = \frac{\rho\Ep[\Delta_A(\Yarr^{arrival})] + \Ep[\Delta_C(\Ycomp^{comp})]+\rho}{1-\rho}
         + \Ep_U[\Delta_A(\Yarr) - \Delta_C(\Ycomp^{comp})],
    \end{align*}
    where $\Ep_U[\cdot]$ was defined in \cref{sec:completion-def-mams}:
    \begin{align*}
        \Ep_U[\Delta_A(\Yarr) - \Delta_C(\Ycomp^{comp}))]  &= \frac{\Ep[(\mu_{\Ycomp} \Delta_A(\Yarr) - \sum_{j_C\in\compSet} s_{\Ycomp, j_C, 1} \Delta_C(j_C)))\indicator{Q=0}]}{\mu-\lambda}.
    \end{align*}
\end{theorem}

Our formula consists of a fully-explicit primary term
and a secondary term (the $\Ep_U[\cdot]$ term) which depends on the behavior of the MAMS system when the queue is empty, captured by the expectation over the unused service $\Ep_U[\cdot]$.
\cref{thm:e-q-exact-mams} generalizes our two-level result, \cref{thm:e-q-exp-two}.
The $\Ep[\Delta(Y) \mid Q = 0]$ term in \cref{thm:e-q-exp-two} becomes the $\Ep_U[\cdot]$ term in \cref{thm:e-q-exact-mams}, and terms are introduced for the Markovian service process.
One can exactly compute the $\Delta_A(\cdot)$ and $\Delta_C(\cdot)$ functions by solving the corresponding Poisson equations, as discussed in \cref{sec:model-rel-arr-mams}.

As one consequence of \cref{thm:e-q-exact-mams}, note that $\Ep_U[\Delta_A(\Yarr) - \Delta_C(\Ycomp^{comp})]$
is bounded by the maximum and minimum possible values of $\Delta_A $ and $\Delta_C$,
yielding the following simple bounds:
\begin{corollary}\label{cor:e-q-bounds-mams}
    Let $\Delta_A^{\max}$ and $\Delta_A^{\min}$ be the maximum and minimum values of $\Delta_A(i_A)$ over all arrival states $i_A$,
    and define $\Delta_C^{\max}$ and $\Delta_C^{\min}$ similarly.  In the MAMS system, the mean queue length is bounded by
        \begin{align*}
        \Delta_A^{\min} - \Delta_C^{\max}\le
        \Ep[Q] - \frac{\rho \Delta_A(\Yarr^{arrival}) + \Delta_C(Y_C^{comp}) + \rho}{1-\rho}
        &\le \Delta_A^{\max} - \Delta_C^{\min}.
    \end{align*}
\end{corollary}
\begin{proof}
    This corollary follows immediately from \cref{thm:e-q-exact-mams}, once we recall that the arrival and completion state spaces $\arrSet$ and $\compSet$ are finite, and that the relative arrivals and completions functions always converge, as discussed in \cref{sec:model-rel-arr-mams}.
\end{proof}

Note that these bounds converge in the $\rho \to 1$ limit, as $\Delta_A^{\max}, \Delta_A^{\min}, \Delta_C^{\max},$ and $\Delta_C^{\min}$ remain bounded in that limit, allowing us to exactly characterize the heavy traffic behavior of the system:

\begin{corollary}\label{cor:heavy-traffic-e-q-mams}
    In the MAMS system, the scaled mean queue length in the heavy traffic limit converges:
    \begin{align*}
        \lim_{\rho\to 1} (1-\rho) \Ep[Q] = \Ep[\Delta_A(\Yarr^{arrival})] + \Ep[\Delta_C(\Ycomp^{comp})]+ 1.
    \end{align*}
\end{corollary}

\section{Proofs for MAMS system}
\label{sec:mams-proofs}

We will follow a similar structure to that used in Section~\ref{sec:two-level-arr}:
\begin{itemize}
    \item \cref{lem:generator-expand} sets up the formula for the drift of our chosen test function $f_{\Delta_A,\Delta_C}$, which we introduced in \cref{sec:test-func},
    splitting the drift into an arrivals term and a completions term.
    \item \cref{lem:arrival-simp,lem:completion-simp} simplify the arrivals and completions drift terms.
    \item \cref{lem:arrival-expectation,lem:completion-expect} compute the steady-state expectations of the arrivals and completions drift terms. 
    \item In \cref{sec:main-result-proof}, we prove \cref{thm:e-q-exact-mams} using our expected drift results and the fact that expected drift in steady state is zero to characterize mean queue length $\Ep[Q]$.
\end{itemize}

\cref{lem:arrival-simp,lem:completion-simp} collectively generalize \cref{lem:G-f-Delta-two-level} from the two-level system to the MAMS system,
while \cref{lem:arrival-expectation,lem:completion-expect,thm:e-q-exact-mams} collectively generalize \cref{thm:e-q-exp-two}.

\subsection{Key lemmas}

We start by deriving an initial formula for the drift of our test function $f_{\Delta_A,\Delta_C}$,
which we introduced in \cref{sec:test-func}.
We separate this formula into two terms: An arrivals term $g_A$ and a completions term $g_C$.

\begin{lemma}[Formula for drift of $f_{\Delta_A,\Delta_C}$]
    For any MAMS system and any system state $(q, i_A, i_C)$,
    \label{lem:generator-expand}
    \begin{align}
    \nonumber
    &G \circ f_{\Delta_A,\Delta_C}(q, i_A, i_C) = g_A(q, i_A, i_C) + g_C(q, i_A, i_C),\\
    \label{eq:g-f-arrival-term}
    &\text{where } g_A(q, i_A, i_C) := \\
    \nonumber
    &\sum_{j_A\in\arrSet, a} r_{i_A, j_A, a} \left((q+a)(q + a + 2\Delta_A(j_A) - 2\Delta_C(i_C)) - q(q + 2\Delta_A(i_A) - 2\Delta_C(i_C))\right), \\
    \label{eq:g-f-completion-term}
    &\text{and where } g_C(q, i_A, i_C) := \\
    \nonumber
    &\sum_{j_C\in\compSet, c} s_{i_C, j_C, c} \left((q-c+u)(q-c+u + 2\Delta_A(i_A) - 2\Delta_C(j_C)) - q(q + 2\Delta_A(i_A) - 2\Delta_C(i_C))\right).  
\end{align}
\end{lemma}
\begin{proof}
    Follows immediately from \cref{lem:g-f-formula-mams}.
\end{proof}

Now, we simplify the $g_A$ term.

\begin{lemma}
    \label{lem:arrival-simp}
    For any MAMS system and any system state $(q, i_A, i_C)$,
    \begin{align}
        g_A(q, i_A, i_C) = 2 \lambda q + (1-2\Delta_C(i_C)) \lambda_{i_A} + 2 \sum_{j_A\in\arrSet} r_{i_A, j_A, 1}\Delta_A(j_A). 
    \end{align}
\end{lemma}
\begin{proof}
We start by simplifying $g_A(q,i_A,i_C)$. Recall the definition of $g_A$, \eqref{eq:g-f-arrival-term}:
\begin{align}
    \label{eq:g-f-arrival-restate}
    &g_A(q, i_A, i_C) = 
    \sum_{j_A, a} r_{i_A, j_A, a} \left((q+a)(q + a + 2\Delta_A(j_A) - 2\Delta_C(i_C)) - q(q + 2\Delta_A(i_A) - 2\Delta_C(i_C))\right).
\end{align}
Let us simplify the bracketed term inside the summation:
\begin{align}
    \nonumber
    &(q+a)(q + a + 2\Delta_A(j_A) - 2\Delta_C(i_C)) - q(q + 2\Delta_A(i_A) - 2\Delta_C(i_C)) \\
    \label{eq:simplify-bracketed}
    &= 2q(a + \Delta_A(j_A) - \Delta_A(i_A)) + a(a-2\Delta_C(i_C)+2\Delta_A(j_A)). 
\end{align}

Let us substitute \eqref{eq:simplify-bracketed} back into \eqref{eq:g-f-arrival-restate}:
\begin{align}
    \label{eq:g-A-q-term}
    g_A(q, i_A, i_C) =\,& 2 q \sum_{j_A, a} r_{i_A, j_A, a} (a + \Delta_A(j_A) - \Delta_A(i_A)) \\
    \label{eq:g-A-a-term}
    &+ \sum_{j_A, a} r_{i_A, j_A, a} a(a-2\Delta_C(i_C)+2\Delta_A(j_A)).
\end{align}

We will separately simplify \eqref{eq:g-A-q-term} and \eqref{eq:g-A-a-term}.

For \eqref{eq:g-A-q-term}, we prove in \cref{cor:effect-arr-comp-on-rel} that for any initial state $i_A$,
\begin{align*}
    \sum_{j_A, a} r_{i_A, j_A, a} (a + \Delta_A(j_A) - \Delta_A(i_A)) = \lambda.
\end{align*}
This follows from the definition of the relative arrivals function $\Delta_A$.

For \eqref{eq:g-A-a-term}, note that $a \in \{0, 1\}$:
\begin{align*}
    &\sum_{j_A, a} r_{i_A, j_A, a} a(a-2\Delta_C(i_C)+2\Delta_A(j_A))
    = (1-2\Delta_C(i_C)) \lambda_{i_A} + 2 \sum_{j_A} r_{i_A, j_A, 1}\Delta_A(j_A).
\end{align*}

Combining our simplifications, we find that
\begin{align*}
    g_A(q, i_A, i_C) &= 2 \lambda q + (1-2\Delta_C(i_C)) \lambda_{i_A} + 2 \sum_{j_A} r_{i_A, j_A, 1}\Delta_A(j_A).
    \qedhere
\end{align*}
\end{proof}

Now, we characterize the expectation of $g_A$ over the steady state of the MAMS system.

\begin{lemma}[Mean of arrivals term $g_A$]
    \label{lem:arrival-expectation}
    For any MAMS system,
    \begin{align}
        \Ep[g_A(Q, \Yarr, \Ycomp)] = \lambda (2\Ep[Q] + 2\Ep[\Delta_A(\Yarr^{arrival})] + 1).
    \end{align}
\end{lemma}
\begin{proof}
    Let us begin with the simplified expression for $g_A$ from \cref{lem:arrival-simp}:
    \begin{align}
    \nonumber
    g_A(q, i_A, i_C) &= 2 \lambda q + (1-2\Delta_C(i_C)) \lambda_{i_A} + 2 \sum_{j_A} r_{i_A, j_A, 1}\Delta_A(j_A). \\
    \label{eq:g-a-expectation-first}
    \Ep[g_A(Q, \Yarr, \Ycomp)] &= 2 \lambda \Ep[Q] + \Ep[(1-2\Delta_C(\Ycomp)) \lambda_{\Yarr}] \\
    &\quad + 2 \sum_{i_A, j_A} P(\Yarr=i_A) r_{i_A, j_A, 1}\Delta_A(j_A).
    \label{eq:g-a-expectation-second}
\end{align}

We simplify \eqref{eq:g-a-expectation-first} and \eqref{eq:g-a-expectation-second} separately. 
For \eqref{eq:g-a-expectation-first}, note that $\Yarr$ and $\Ycomp$ are independent,
as the arrival process and completions process update independently.
Note also that $\Ep[\Delta_C(\Ycomp))]=0$, as shown in \cref{lem:mean-delta},
and that $\Ep[\lambda_{\Yarr}] = \lambda$, which is the definition of $\lambda$.
Thus, we can simplify the second term of \eqref{eq:g-a-expectation-first}:
\begin{align}
    \label{eq:g-a-simplified-first}
    \Ep[(1-2\Delta_C(\Ycomp)) \lambda_{\Yarr}] = \Ep[(1-2\Delta_C(\Ycomp))] \Ep[\lambda_{\Yarr}]
    &= \lambda.
\end{align}

As for \eqref{eq:g-a-expectation-second},
recall from \cref{sec:arrival-def-mams} the definition of $\Yarr^{arrival}$, the arrival-rate-weighted arrival state:
\begin{align*}
    P(\Yarr^{arrival} = j) := \frac{\sum_{i\in\arrSet} P(\Yarr = i) r_{i,j,1}}{\lambda}.
\end{align*}

As a result, we can simplify \eqref{eq:g-a-expectation-second}:
\begin{align}
    \sum_{i_A, j_A} P(\Yarr=i_A) r_{i_A, j_A, 1}\Delta_A(j_A)
    &= \lambda \sum_{j_A} P(\Yarr^{arrival} = j_A) \Delta_A(j_A)
    \label{eq:g-a-simplified-second}
    = \lambda \Ep[\Delta_A(\Yarr^{arrival})].
\end{align}
Substituting \eqref{eq:g-a-simplified-first} and \eqref{eq:g-a-simplified-second} into \eqref{eq:g-a-expectation-first}, we find that
\begin{align*}
    \Ep[g_A(Q, \Yarr, \Ycomp)] &= \lambda (2 \Ep[Q] + 2 \Ep[\Delta_A(\Yarr^{arrival})] + 1).
    \qedhere
\end{align*}
\end{proof}

Now, we switch to the completions term, $g_C$.
We start by simplifying $g_C$.
\cref{lem:completion-simp} mirrors \cref{lem:arrival-simp},
but with some additional complications due to unused service when $q=0$.

\begin{lemma}[Simplification of the completion term $g_C$]
For any MAMS system and any system state $(q, i_A, i_C)$,
    \label{lem:completion-simp}
    \begin{align}
        g_C(q, i_A, i_C) = -2\mu q + \Big((1 - 2\Delta_A(i_A))\mu_{i_C} + 2 \sum_{j_C\in\compSet} s_{i_C, j_C, 1} \Delta_C(j_C)\Big) \indicator{q > 0}. 
    \end{align}
\end{lemma}

\begin{proof}
    Recall the definition of $g_C(q, i_A, i_C)$ from \cref{lem:g-f-formula-mams}: 
    \begin{align*}
        & g_C(q, i_A, i_C) := 
         \sum_{j_C, c} s_{i_C, j_C, c} \left((q-c+u)(q-c+u + 2\Delta_A(i_A) - 2\Delta_C(j_C)) - q(q + 2 \Delta_A(i_A) - 2 \Delta_C(i_C))\right),  
    \end{align*}
     where $u := (c-q)^+$ is the amount of unused service caused by a specific transition. Note that $u \in \{0, 1\}$.

     Let us simplify the bracketed term inside the summation:
     \begin{align}
        \nonumber
         &(q-c+u)(q-c+u + 2\Delta_A(i_A) - 2\Delta_C(j_C)) - q(q + 2 \Delta_A(i_A) - 2 \Delta_C(i_C))  \\
         \nonumber
         \nonumber
         &= 2q(-c+u - \Delta_C(j_C) + \Delta_C(i_C)) + (-c+u)(-c+u+2\Delta_A(i_A) - 2\Delta_C(j_C)) \\
         &= -2q(c + \Delta_C(j_C) - \Delta_C(i_C)) + c(c - 2\Delta_A(i_A) + 2\Delta_C(j_C))
         \label{eq:g-c-bracket-simplify}
         + u (2q -2c + u + 2\Delta_A(i_A) - 2\Delta_C(j_C)). 
     \end{align}
     We substitute \eqref{eq:g-c-bracket-simplify} back to the definition of $g_C(q, i_A, i_C)$ and get 
     \begin{align}
         g_C(q, i_A, i_C) 
         \label{eq:gc-decomp-q-term}
         &= -2q \sum_{j_C, c} s_{i_C, j_C, c} (c + \Delta_C(j_C) - \Delta_C(i_C)) \\
         \label{eq:gc-decomp-c-term}
         &+ \sum_{j_C, c} s_{i_C, j_C, c} c(c - 2\Delta_A(i_A) + 2\Delta_C(j_C)) \\
         \label{eq:gc-decomp-u-term}
         &+ \sum_{j_C, c} s_{i_C, j_C, c} u (2q -2c + u + 2\Delta_A(i_A) - 2\Delta_C(j_C)). 
     \end{align}

     For the term in \eqref{eq:gc-decomp-q-term}, we prove in \cref{cor:effect-arr-comp-on-rel} that 
     \[
        \sum_{j_C, c} s_{i_C, j_C, c} (c +   \Delta_C(j_C) - \Delta_C(i_C)) = \mu. 
     \]
     so the term in \eqref{eq:gc-decomp-q-term} can be simplified as 
    \begin{align*}
        -2q \sum_{j_C, c} s_{i_C, j_C, c} (c + \Delta_C(j_C) - \Delta_C(i_C)) = -2\mu q. 
    \end{align*}

    For the term in \eqref{eq:gc-decomp-c-term}, note that $c \in \{0,1\}$, so 
    \begin{align*}
         \sum_{j_C, c} s_{i_C, j_C, c} c(c - 2\Delta_A(i_A) + 2\Delta_C(j_C))
         &= \sum_{j_C} s_{i_C, j_C, 1} (1 - 2\Delta_A(i_A) + 2\Delta_C(j_C)) \\
         &= (1 - 2 \Delta_A(i_A))\mu_{i_C} + 2 \sum_{j_C} s_{i_C, j_C, 1} \Delta_C(j_C). 
    \end{align*}

    For the term in \eqref{eq:gc-decomp-u-term}, recall that $u = (c-q)^+$, so $u\in\{0,1\}$, and $u = 1$ if and only if $c = 1$ and $q = 0$. We thus simplify the term in \eqref{eq:gc-decomp-u-term} as 
    \begin{align*}
        &\sum_{j_C, c} s_{i_C, j_C, c} u (2q -2c + u + 2\Delta_A(i_A) - 2\Delta_C(j_C)) 
        \\
        &= \sum_{j_C} s_{i_C, j_C, 1} (-1 + 2\Delta_A(i_A) - 2\Delta_C(j_C)) \indicator{q = 0} \\
        &= - \Big((1 - 2\Delta_A(i_A))\mu_{i_C}  + 2 \sum_{j_C} s_{i_C, j_C, 1} \Delta_C(j_C)\Big) \indicator{q = 0}. 
    \end{align*}

    Combining our simplifications, we find that 
    \[
        g_C(q, i_A, i_C)  = -2\mu q + \Big((1 - 2\Delta_A(i_A))\mu_{i_C}  + 2 \sum_{j_C} s_{i_C, j_C, 1} \Delta_C(j_C)\Big) \indicator{q > 0}. 
        \qedhere
    \]
\end{proof}

\begin{lemma}[Mean of completions term $g_C$]
    \label{lem:completion-expect}
    For any MAMS system,
    \begin{align*}
        &\Ep[g_C(Q, Y_A, Y_C)]
        =  -2\mu \Ep[Q] + \lambda + 2\mu \Ep[\Delta_C(\Ycomp^{comp})] + 2 \Ep[h_{C0}(\Yarr, \Ycomp) \indicator{Q=0}],
    \end{align*}
    where $h_{C0}$ is defined as follows:
    \begin{align*}
        h_{C0}(i_A, i_C) &= \mu_{i_C} \Delta_A(i_A) - \sum_{j_C\in\compSet} s_{i_C, j_C, 1} \Delta_C(j_C), \\
        \Ep[h_{C0}(\Yarr, \Ycomp) \indicator{Q=0}] &= (\mu-\lambda)\Ep_U[\Delta_A(\Yarr) - \Delta_C(\Ycomp^{comp})].
    \end{align*}
\end{lemma}

\begin{proof}
    By the simplified formula of $g_C(q, i_A, i_C)$ in  \Cref{lem:completion-simp}, we have
    \begin{equation}
        \label{eq:g-c-expect-intermediate}
        \Ep[g_C(Q, \Yarr, \Ycomp)]
        =  -2\mu \Ep[Q] + \Ep[ \big((1 - 2\Delta_A(\Yarr))\mu_{\Ycomp} + 2 \sum_{j_C} s_{\Ycomp, j_C, 1} \Delta_C(j_C)\big) \indicator{Q > 0}].
    \end{equation}
    We can decompose the second term of \eqref{eq:g-c-expect-intermediate} using the fact that $\indicator{Q>0} = 1 - \indicator{Q=0}$. 
    Observe that
    \begin{align*}
        \Ep[(1 - 2\Delta_A(\Yarr))\mu_{\Ycomp} + 2 \sum_{j_C} s_{\Ycomp, j_C, 1} \Delta_C(j_C)] 
        &= \Ep[1-2\Delta_A(\Yarr)] \Ep[\mu_{\Ycomp}] + 2 \mu \Ep[\Delta_C(\Ycomp^{comp})] \\
        &= \mu + 2\mu \Ep[\Delta_C(\Ycomp^{comp})],
    \end{align*}
    where the first equality uses the independence of $\Yarr$ and $\Ycomp$, and the definition of $\Ycomp^{comp}$; and the second equality is due to $\Ep[\mu_{\Ycomp}] = \mu$ and $\Ep[\Delta_A(\Yarr)] = 0$. 
    Next, we deal with the term
    \begin{align*}
        &\Ep[ \big((1 - 2\Delta_A(\Yarr))\mu_{\Ycomp} + 2 \sum_{j_C} s_{\Ycomp, j_C, 1} \Delta_C(j_C)\big) \indicator{Q = 0}]  \\
        &= \Ep[\mu_{\Ycomp}\indicator{Q=0}] + \Ep[ \big( - 2\Delta_A(\Yarr)\mu_{\Ycomp} + 2 \sum_{j_C} s_{\Ycomp, j_C, 1} \Delta_C(j_C)\big) \indicator{Q = 0}] \\
        &= \mu - \lambda  -2 \Ep[h_{C0}(\Yarr, \Ycomp) \indicator{Q=0}],
    \end{align*} 
    where the last equality is due to $\Ep[\mu_{\Ycomp}\indicator{Q=0}] = \mu - \lambda$, which we prove in \Cref{lem:unused}. 
    Plugging the above calculations into \eqref{eq:g-c-expect-intermediate}, we get 
    \begin{align}
        \Ep[g_C(Q, \Yarr, \Ycomp)] 
        \nonumber
        &= -2\mu \Ep[Q] + \mu + 2\mu \Ep[\Delta_C(\Ycomp^{comp})] - \mu + \lambda  + 2 \Ep[h_{C0}(\Yarr, \Ycomp) \indicator{Q=0}] \\
        \nonumber
        &= -2\mu \Ep[Q] + \lambda + 2\mu \Ep[\Delta_C(\Ycomp^{comp})] + 2 \Ep[h_{C0}(\Yarr, \Ycomp) \indicator{Q=0}]. 
    \end{align} 
    Equivalently, we can write the last term $\Ep\left[h_{C0}(\Yarr, \Ycomp)\indicator{Q=0}\right]$ in terms of $\Ep_U[\cdot]$:
    \begin{align*}
        \Ep\left[h_{C0}(\Yarr, \Ycomp)\indicator{Q=0}\right]
        &= \Ep[\mu_{\Ycomp} \Delta_A(\Yarr) \indicator{Q=0}]- \Ep[\sum_{j_C} s_{\Ycomp, j_C, 1} \Delta_C(j_C))\indicator{Q=0}] \\
        &= (\mu-\lambda)\Ep_U[\Delta_A(\Yarr)] - (\mu-\lambda)\Ep_U[\Delta_C(\Ycomp^{comp})] \\
        &= (\mu-\lambda)\Ep_U[\Delta_A(\Yarr) - \Delta_C(\Ycomp^{comp})]. \qedhere
    \end{align*}
\end{proof}

\subsection{Proof of main result: \cref{thm:e-q-exact-mams}}\label{sec:main-result-proof}

\begin{reptheorem}{thm:e-q-exact-mams}
    In the MAMS system, the mean queue length is given by:
    \begin{align*}
        &\Ep[Q] = \frac{\rho\Ep[\Delta_A(\Yarr^{arrival})] + \Ep[\Delta_C(\Ycomp^{comp})]+\rho}{1-\rho}
         + \Ep_U[\Delta_A(\Yarr) - \Delta_C(\Ycomp^{comp})],
    \end{align*}
    where recall the definition of $\Ep_U[\cdot]$ in \Cref{sec:completion-def-mams} that 
    \begin{align*}
        \Ep_U[\Delta_A(\Yarr)- \Delta_C(\Ycomp^{comp}))]  &= \frac{\Ep[(\mu_{\Ycomp} \Delta_A(\Yarr) - \sum_{j_C\in\compSet} s_{\Ycomp, j_C, 1} \Delta_C(j_C)))\indicator{Q=0}]}{\mu-\lambda}.
    \end{align*}
\end{reptheorem}

\begin{proof}
    Recall \Cref{lem:generator-expand}, $
        G \circ f_{\Delta_A,\Delta_C}(q, i_A, i_C) = g_A(q, i_A, i_C) + g_C(q, i_A, i_C)
    $.
    Because $\Ep[G\circ f_{\Delta_A,\Delta_C}(Q, \Yarr, \Ycomp)] = 0$ by \cref{lem:drift-lemma-mams}, we have
    \[
        \Ep[g_A(Q, \Yarr \Ycomp)] + \Ep[g_C(Q, \Yarr, \Ycomp)] = 0.
    \]
    Using \Cref{lem:arrival-expectation,lem:completion-expect}, we get
    \begin{align*}
        & \lambda (2\Ep[Q] + 2\Ep[\Delta_A(\Yarr^{arrival})] + 1)\\
        &-2\mu \Ep[Q] + \lambda + 2\mu \Ep[\Delta_C(\Ycomp^{comp})] + 2(\mu-\lambda)\Ep_U[\Delta_A(\Yarr) - \Delta_C(\Ycomp^{comp})] = 0. 
    \end{align*}
    Rearranging the terms, 
    \begin{align*}
        2(\mu - \lambda) \Ep[Q] &= 2\lambda \Ep[\Delta_A(\Yarr^{arrival})] + 2\mu \Ep[\Delta_C(\Ycomp^{comp})] + 2\lambda + 2(\mu-\lambda)\Ep_U[\Delta_A(\Yarr) - \Delta_C(\Ycomp^{comp})] \\
        \Ep[Q] &= 
        \frac{\lambda \Ep[\Delta_A(\Yarr^{arrival})] + \mu \Ep[\Delta_C(\Ycomp^{comp})] + \lambda}{\mu-\lambda} + \Ep_U[\Delta_A(\Yarr) - \Delta_C(\Ycomp^{comp})] \\
        &= \frac{\rho \Delta_A(\Yarr^{arrival}) + \Delta_C(\Ycomp^{comp})
        +\rho}{1-\rho} + \Ep_U[\Delta_A(\Yarr) - \Delta_C(\Ycomp^{comp})].
        \qedhere
    \end{align*}
\end{proof}

\section{Bounds on two-level system when $Q=0$}
\label{sec:two-level-bounds}

We proved \cref{thm:e-q-exact-mams}, a general result characterizing mean queue length in the MAMS system.
As a corollary, in \cref{cor:e-q-bounds-mams}, we proved explicit bounds on mean queue length in the MAMS system. 
In the two-level arrival systems, stronger results can be obtained because we know more about the dynamics of the system. Specifically, recall that we proved \cref{thm:e-q-exp-two}, which characterized mean queue length $\Ep[Q]$, with the only remaining uncertainty being the probability $P(Y=H \mid Q = 0)$ that the system is in the $H$ arrival state while the queue is empty. 

In \cite{yechiali_queuing_1971}, it is shown that the exact value of this probability can be obtained by solving a cubic equation, which leads to the exact value of the mean queue length. 
In this section, we focus on proving explicit upper bounds on $P(Y=H \mid Q = 0)$, as stated in \cref{thm:two-level-bounds}. These bounds complement the straightforward lower bound $P(Y=H \mid Q = 0) \ge 0$. 

\begin{theorem}
    \label{thm:two-level-bounds}
    In the two-level arrival system, $P(Y=H \mid Q=0)$
    is bounded as follows:
    \begin{align}
        \label{eq:rep-fast-bound}
        P(Y = H \mid Q=0) &\le P(Y=H) = \frac{\alpha_L}{\alpha_L+\alpha_H}.
    \end{align}
    Additionally, if the system is intermittently overloaded ($\lambda_H > \mu$),
    the following bound holds:
    \begin{align}
        \label{eq:rep-slow-bound}
        P(Y = H \mid Q=0) &\le \frac{1}{1-\rho}\frac{1}{\lambda_H - \mu}\frac{\alpha_H \alpha_L}{\alpha_H + \alpha_L}.
    \end{align}
\end{theorem}

We will prove \cref{thm:two-level-bounds} in \cref{sec:proof-two-level}.

These bounds are tight in different asymptotic regimes: \eqref{eq:rep-fast-bound} is tight in the $\alpha_H,\alpha_L \to \infty$ limit, where the arrival state switches rapidly,
while \eqref{eq:rep-slow-bound} 
is tight in the $\alpha_H,\alpha_L \to 0$ limit, as the arrival state switches slowly.
These bounds complement our heavy-traffic result for the two-level system, \cref{cor:heavy-traffic-e-q-two-level}.

Using these bounds alongside \cref{thm:e-q-exp-two},
we can derive tight bounds on mean queue length in the two-level arrival system:
\begin{corollary}
    \label{cor:two-level-e-q-bounds}
    In the two-level arrivals system  with exponential service rate,
    the mean queue length is bounded as follows:
    \begin{align}
        \label{eq:two-level-e-q-lower}
        \Ep[Q] &\ge \frac{\rho}{1-\rho}\left( \frac{(\lambda_H - \lambda_L)^2 \alpha_L \alpha_H}{\lambda (\alpha_H+\alpha_L)^3} + 1 \right)
        - \frac{\lambda_H - \lambda_L}{\alpha_H + \alpha_L} \frac{\alpha_L}{\alpha_H + \alpha_L}, \\
        \label{eq:two-level-e-q-upper-fast}
        \Ep[Q] &\le \frac{\rho}{1-\rho}\left( \frac{(\lambda_H - \lambda_L)^2 \alpha_L \alpha_H}{\lambda (\alpha_H+\alpha_L)^3} + 1 \right).
    \end{align}
    If the system is intermittently overloaded ($\lambda_H > \mu$), then additionally
    \begin{align}
        \label{eq:two-level-e-q-upper-slow}
        \Ep[Q] &\le \frac{\rho}{1-\rho}\left( \frac{(\lambda_H - \lambda_L)^2 \alpha_L \alpha_H}{\lambda (\alpha_H+\alpha_L)^3} + 1 \right)
        - \frac{\lambda_H - \lambda_L}{\alpha_H + \alpha_L} \frac{\alpha_L}{\alpha_H + \alpha_L}\left( 1
        -  \frac{\alpha_H}{(1-\rho)(\lambda_H - \mu)} \right).
    \end{align}
\end{corollary}
\begin{proof}
    These bounds follow from combining \cref{thm:e-q-exp-two}
    with bounds on $P(Y=H \mid Q=0)$.
    \eqref{eq:two-level-e-q-lower} follows from the fact that $P(Y=H \mid Q=0) \ge 0$.
    \eqref{eq:two-level-e-q-upper-fast} follows from  \eqref{eq:rep-fast-bound}.
    \eqref{eq:two-level-e-q-upper-slow} follows from \eqref{eq:rep-slow-bound}.
\end{proof}

In the non-intermittently-overloaded case ($\lambda_H < \mu$),
prior work has proven basic bounds on $\Ep[Q]$.
\citet{gupta_fundamental_2006} proved that $\Ep[Q]$ is bounded below by
the mean queue length of a M/M/1
with arrival rate $\lambda$ and completion rate $\mu$,
and bounded above by a weighted mixture of two M/M/1s: one with arrival rate $\lambda_H$ and one with arrival rate $\lambda_L$.
While these bounds are not tight, it is challenging to use our techniques to prove tighter results in this setting, which we leave to future work.

\subsection{Proof of \cref{thm:two-level-bounds}}
\label{sec:proof-two-level}
\begin{proof}

We start by considering the general two-level arrival system, with no restriction on $\lambda_H$ and $\mu$. We will prove \eqref{eq:rep-fast-bound}:
\begin{align*}
    P(Y = H \mid Q=0) &\le P(Y=H) = \frac{\alpha_L}{\alpha_L+\alpha_H}.
\end{align*}

To prove this result, we make use of prior work by \citet{gupta_fundamental_2006}
    to bound the probability $P(Q=0 \mid Y=L)$,
    the probability that the queue is empty given that the arrival rate is low.

    Let $\pi_0^L$ be the probability that the queue is empty at the moment when the system switches from $L$ to $H$.
    By \cite[Theorem 3]{gupta_fundamental_2006},
    $P(Q=0 \mid Y=L) = \pi_0^L$.
    During the proof of \cite[Theorems 6 and 7]{gupta_fundamental_2006},
    it is proven that there exists some $\chi \in (0, 1)$ such that
        $\pi_0^L / \pi_0^H = (\mu-\chi \lambda_L) / (\mu-\chi \lambda_H).$

    Because $\lambda_L \le \lambda_H$, we know that $\pi_0^L \ge \pi_0^H$,
    and hence that $P(Q=0 \mid Y=L) \ge P(Q=0 \mid Y=H)$.

    To quantify $P(Q=0)$, note that in all systems, the fraction of potential service events which go unused is $1-\rho$.
In the two-level arrival system, because the completion process is exponential,
it is therefore the case that $P(Q=0) = 1-\rho$.
    
    Thus, we know that
    $P(Q=0 \mid Y = L) \ge 1 - \rho$,
    and that $P(Q=0 \mid Y = H) \le 1 - \rho$.
    We can therefore bound $P(Y=H \mid Q = 0)$:
    \begin{align*}
        P(Y=L \mid Q = 0) &\ge P(Y = L) = \frac{\alpha_H}{\alpha_H + \alpha_L},
        \quad
        P(Y=H \mid Q = 0) \le P(Y = H) = \frac{\alpha_L}{\alpha_H + \alpha_L}.
    \end{align*}

Now, consider the intermittently overloaded two-level arrival system, where $\lambda_H > \mu$.
We will prove \eqref{eq:rep-slow-bound}:
\begin{align*}
    P(Y = H \mid Q=0) &\le \frac{1}{1-\rho}\frac{1}{\lambda_H - \mu}\frac{\alpha_H \alpha_L}{\alpha_H + \alpha_L}.
\end{align*}

The intuition behind this bound is that if $\lambda_H > \mu$,
the queue length grows throughout each $Y = H$ interval,
and hence the queue is rarely empty ($Q$ is rarely $0$) during such an interval.

To make this idea rigorous, consider an overloaded M/M/1 with arrival rate $\lambda_H$ greater than the service rate $\mu$, in which the queue is empty at time 0 ($Q(0) = 0$).
The expected total amount of time for which $Q(t) = 0$ in the overloaded M/M/1 system is finite: by a standard calculation, one can show that this expected duration is $1/(\lambda_H - \mu)$.
That amount of time is an upper bound on the expected amount of time for which $Q=0$
during each $Y = H$ interval in the two-level system.





Thus, the expected time that two-level system spends in the state $y = H, q=0$ during each $y=H$ interval is at most $1/(\lambda_H - \mu)$, regardless of the initial queue length at the beginning of the interval.
A larger initial queue length and a finite interval length can only decrease the expected time for which the system is empty.

To bound $P(Y = H \land Q = 0)$, we use a renewal-reward argument, with two kinds of cycles:
A \emph{switching cycle}, consisting of a $y = H$ interval followed by a $y = L$ interval,
and a \emph{zero-switching cycle}, which is a renewal cycle with a renewal point whenever the system
changes from $y=L$ to $y = H$ at a moment when $Q=0$.
Let $Z$ be a random variable denoting the number of switching cycles within each zero-switching cycle. 
Note that $Z$ has a finite mean because the two-level system is positive recurrent with $\lambda < \mu$. 

We now apply the renewal-reward theorem \citep{harchol_performance_2013}, with the renewal period being the zero-switching cycles, and the reward being
the time spent in the state $Y = H, Q=0$:
\begin{align*}
    P(Y = H \land Q = 0) &= \frac{\Ep[\text{$Y=H \land Q=0$ duration during a zero-switching cycle}]}{\Ep[\text{Length of a zero-switching cycle}]}. 
\end{align*}
The lengths of the switching cycles within the zero-switching cycles are i.i.d. with expected length $\frac{1}{\alpha_H} + \frac{1}{\alpha_L}$, 
and the length of zero-switching cycle is a stopping time with mean $E[Z](\frac{1}{\alpha_H} + \frac{1}{\alpha_L})$.
The amount of time for which $Y=H \land Q=0$
during the switching cycles are correlated, but only via the length of the queue at the start of the cycle. Regardless of that queue length, the expected time for which $Y=H \land Q=0$ per switching cycle is at most $1/(\lambda_H - \mu)$. As a result,
\begin{align*}
    P(Y = H \land Q = 0) &\le \frac{\Ep[Z]/(\lambda_H - \mu)}{\Ep[Z](1/\alpha_H + 1/\alpha_L)}
    = \frac{1/(\lambda_H - \mu)}{1/\alpha_H + 1/\alpha_L}. 
\end{align*}

Recall that $P(Q=0) = 1-\rho$, because the system has exponential service.
As a result,
\begin{align*}
    P(Y = H \mid Q=0) \le \frac{1}{1-\rho}\frac{1/(\lambda_H - \mu)}{1/\alpha_H + 1/\alpha_L} = \frac{1}{1-\rho}\frac{1}{\lambda_H - \mu} \frac{\alpha_H \alpha_L}{\alpha_H + \alpha_L}. \quad
    \qedhere
\end{align*}
\end{proof}

Note that \eqref{eq:rep-slow-bound} scales linearly with $\alpha_H,\alpha_L$, so it converges to 0 in the $\alpha_H,\alpha_L \to 0$ limit.

\section{Simulation}
\label{sec:simulation}

\begin{figure}
    \centering
    \includegraphics[width=0.8\textwidth]{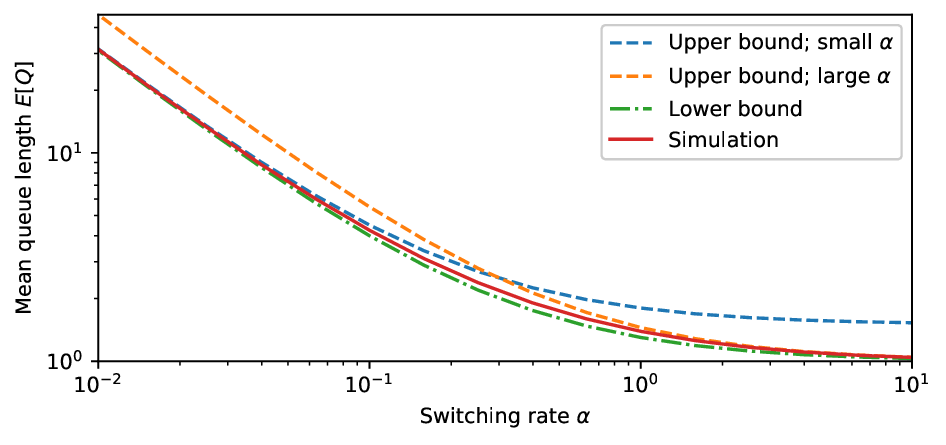}
    \caption{Simulation and bounds for mean queue length $\Ep[Q]$ in two-level system with intermittent overload,
    under varying switching rates $\alpha_H, \alpha_L$.
    Setting: $\lambda_H = 2, \lambda_L=0.2, \alpha_H = 5 \alpha_L, \mu=1, \rho=0.5$. Overall switching rate $\alpha := (\alpha_H + \alpha_L)/2$.
    Bounds given in \cref{cor:two-level-e-q-bounds}: The upper bound which is tight for small $\alpha$ is \eqref{eq:two-level-e-q-upper-slow}, for large $\alpha$ is \eqref{eq:two-level-e-q-upper-fast}, and the lower bound is \eqref{eq:two-level-e-q-lower}.
    Simulated $10^8$ arrivals.}
    \label{fig:bounds-alpha}
\end{figure}

We have characterized the mean queue length for the MAMS system, with bounds given in \cref{cor:e-q-bounds-mams}. We have proven even tighter bounds for the two-level arrival system in \cref{cor:two-level-e-q-bounds}.

In this section, we simulate mean queue length for the two-level arrival system and a more general MAMS system, and compare our simulations to our bounds, to validate our bounds and demonstrate their tightness.

In \cref{fig:bounds-alpha}, we simulate an intermittently overloaded two-level arrivals system under a variety of switching rates $\alpha_H, \alpha_L$.
In \cref{fig:bounds-alpha}, we have moderate overall load, because $\rho=0.5$,
but significant periods of intermittent overload, because $\lambda_H = 2 > \mu = 1$.

\cref{fig:bounds-alpha} shows that our bounds in \cref{cor:two-level-e-q-bounds}
are tight in both the fast-switching limit (large $\alpha$, short periods of time in each arrival state) and the slow-switching limit (small $\alpha$, long periods of time in each arrival state),
with our bounds tightly and accurately constraining the mean queue length in both limits.
Specifically, \eqref{eq:two-level-e-q-lower} and \eqref{eq:two-level-e-q-upper-fast} are tight in the fast-switching limit, while \eqref{eq:two-level-e-q-lower} and \eqref{eq:two-level-e-q-upper-slow} are tight in the slow-switching limit.
Moreover, our bounds are reasonably tight at all switching rates $\alpha$, even outside the asymptotic regime, with a maximum absolute error of $0.5$ jobs,
and a maximum relative error of 23\%.

\begin{figure}
    \centering
    \includegraphics[width=\textwidth]{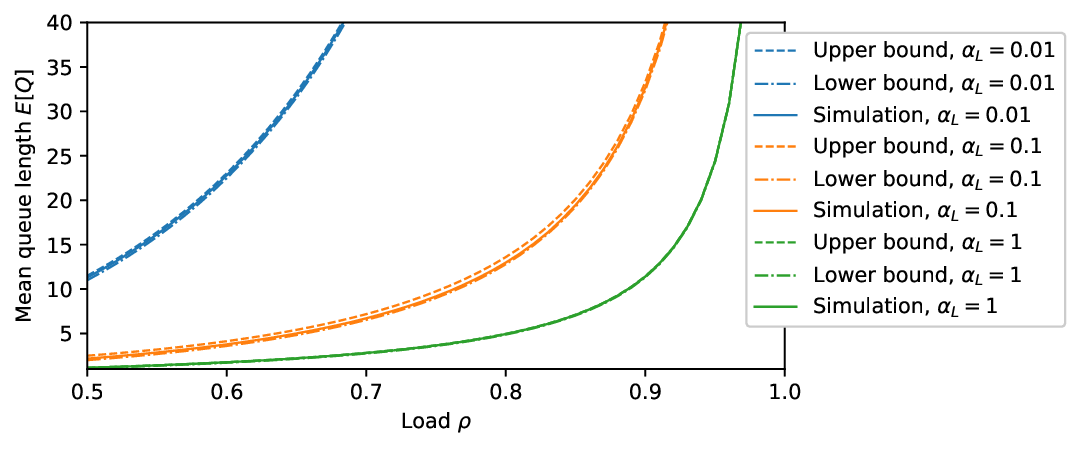}
    \caption{Simulation and bounds for mean queue length $\Ep[Q]$ in two-level system with intermittent overload,
    under varying load $\rho$, for three different values of $\alpha_L$.
    Setting: $\lambda_H = 4 \rho, \lambda_L=0.4 \rho, \alpha_H = 5\alpha_L, \mu=1$.
    Bounds given in \cref{cor:two-level-e-q-bounds}: Upper bound is minimum of \eqref{eq:two-level-e-q-upper-fast} and \eqref{eq:two-level-e-q-upper-slow}. Simulated $10^8$ arrivals.}
    \label{fig:bounds-rho-two}
    \includegraphics[width=0.8\textwidth]{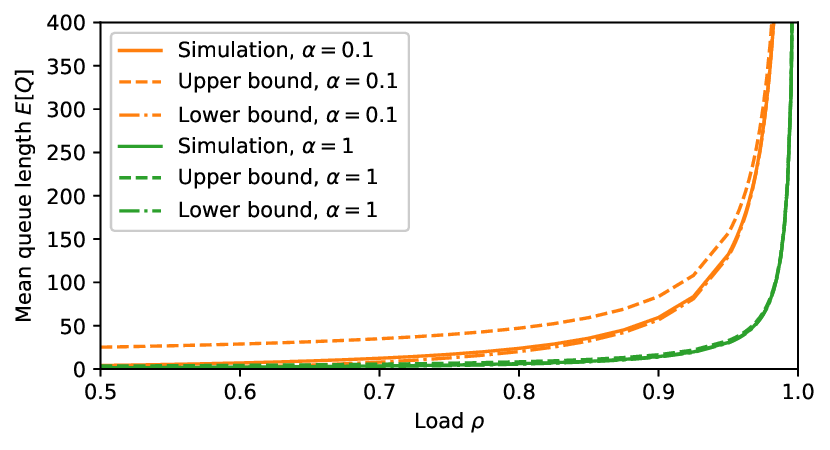}
    \caption{
    Setting: MAMS queue with three arrivals levels: $[0.3\rho, 2\rho, 2.2\rho]$, and three completions levels: [0.5, 1.0, 3.0].
    The system remains in each arrival state and each service state for time $Exp(\alpha)$, then moves cyclically to the next rate in the list, wrapping around. Bounds given in \cref{cor:e-q-bounds-mams}. Simulated $10^9$ arrivals.}
    \label{fig:bounds-rho-mams}
\end{figure}

In \cref{fig:bounds-rho-two}, we simulate a two-level system under a variety of loads $\rho$, and under three different pairs of switching rates $\alpha_H,\alpha_L$.
We compare the simulated mean response time against our bounds from \cref{cor:two-level-e-q-bounds}.
In every configuration of load $\rho$ and switching rates $\alpha_H,\alpha_L$ simulated in \cref{fig:bounds-rho-two}, the absolute gap between our upper and lower bounds is at most 1 job. The relative error between our bounds drops towards zero as the system moves towards heavy traffic ($\rho \to 1$), in each setting.

In \cref{fig:bounds-rho-mams}, we simulate a MAMS system with three levels of arrival rates and three levels of completion rates, under varying loads $\rho$.
The system cycles between three possible arrival rates and three possible completion rates.
Recall that as with all MAMS systems, the arrival chain and completions chain evolve independently.
The switching rate from each arrival and completion state to the next is a fixed rate $\alpha$.
In \cref{fig:bounds-rho-mams}, we show two switching rates $\alpha$.
We compare the simulated mean queue lengths $\Ep[Q]$ against our bounds from \cref{cor:e-q-bounds-mams}.
\cref{fig:bounds-rho-mams} illustrates that our bounds are tight in heavy traffic, even for the more complex MAMS system, confirming the heavy-traffic limit result in \cref{cor:heavy-traffic-e-q-mams}. 
Unlike the two-level arrival setting illustrated in \cref{fig:bounds-alpha,fig:bounds-rho-two}, our bounds are not as tight in the general MAMS setting in light traffic (i.e., for small $\rho$), especially for slower switching settings (low $\alpha$).
Proving more precise bounds on $\Ep[Q]$ for general MAMS systems in light traffic is left for future work. 

\section{Conclusion}

We analyze the Markovian arrivals Markovian service (MAMS) system,
using a somewhat novel framework based on the drift method and the concepts of relative arrivals and relative completions.
While the concepts of relative arrivals and relative completions are a rediscovery of existing techniques \cite{falin_heavy_1999}, their application via the drift method \cite{eryilmaz_asymptotically_2012} is novel to this work.
Using this framework, we capture the impact of the correlated arrivals and completions on mean queue length.
We derive the explicit bounds on mean queue length in the MAMS system,
with bounds that are tight in heavy traffic, matching prior results.

Moreover, in the important special case of the two-level arrivals system under intermittent overload, we prove significantly stronger bounds which are tight both in the fast-switching and slow-switching limits,
again matching prior results.

An important direction for future work is to analyze
MAMS systems with infinite arrival and/or completions chains.
Such systems do not appear to have been studied in prior work.
Our proof of the main theorem, \cref{thm:e-q-exact-mams}, holds in such systems,
assuming appropriate conditions on the arrivals and completions chains,
but it does not immediately imply a tight bound on mean queue length.

\section{Acknowledgements}

Isaac Grosof was supported by the 2024 Tennenbaum Postdoctoral Fellowship from the School of Industrial and Systems Engineering at Georgia Institute of Technology,
and by Air Force Office of Scientific Research Grant FA9550-24-1-0002.
Yige Hong was supported by National Science Foundation Grant NSF-ECCS-2145713.
Mor Harchol-Balter was supported by National Science Foundation Grants NSF-IIS-2322973, NSF-CMMI-1938909, and NSF-CMMI-2307008.

\bibliographystyle{plainnat}
\bibliography{refs}

\appendix
\section{Expectation over unused service}
\label{app:unused}

We now formally define the expectation $\Ep_U[\cdot]$, the \emph{expectation at moments of unused service}.
Unused service occurs at rate $\mu_{i_C}$ whenever $q=0$.

For any test function $f(q, i_A, i_C)$, we define its expectation at moments of unused service $\Ep_U[\cdot]$ as 
\begin{align*}
    \Ep_U[f(Q, \Yarr, \Ycomp)] & := \frac{\Ep[\mu_{\Ycomp} f(Q, \Yarr, \Ycomp)\indicator{Q=0}]}{\Ep[\mu_{\Ycomp}\indicator{Q=0}]} = \frac{\Ep[\mu_{\Ycomp} f(Q, \Yarr, \Ycomp)\indicator{Q=0}]}{\mu-\lambda}, 
\end{align*}
where the second equality holds because $\Ep[\mu_{\Ycomp} \indicator{Q=0}] = \mu - \lambda$ by \cref{lem:unused}.

In the special case of exponential service times, $\mu_{i_C} = \mu$  for any $i_C\in \compSet$, so
\begin{align*}
    \Ep_U[f(Q, \Yarr, \Ycomp)] &= \Ep[f(Q, \Yarr, \Ycomp) \mid Q=0] = \frac{\Ep[f(Q, \Yarr, \Ycomp)\indicator{Q=0}]}{1-\rho}.
\end{align*}

We will also write $\Ep_U[f(Y_C^{comp})]$, which is defined as the expectation of $f(i_C)$ with respect to the distribution of $\Ycomp$ right after unused service, i.e., the completions that happen during $Q=0$. Formally, 
\begin{align*}
    \Ep_U[f(Y_C^{comp})] & := \Ep_U\Bigg[\sum_{j_C\in\compSet} \frac{s_{\Ycomp, j_C, 1}}{\mu_{\Ycomp}}f(j_C)\Bigg] 
    = \frac{\Ep \left[\sum_{j_C\in\compSet} s_{\Ycomp, j_C, 1}f(j_C)\indicator{Q=0}\right]}{\Ep\left[\mu_{\Ycomp}\indicator{Q=0} \right]}.
\end{align*}
Note the similarity with the definition of $Y_C^{comp}$ in \cref{sec:completion-def-mams}.

We now characterize the rate of \emph{unused service} in the system, namely completion transitions
that occur when the queue is empty:
\begin{lemma}
    \label{lem:unused}
    The rate of unused service is $\mu - \lambda$: $\Ep_{i \sim \Ycomp}[\mu_i \indicator{Q=0}] = \mu - \lambda$.
\end{lemma}
\begin{proof}
    The long-run arrival rate is $\lambda$.
    The long-run rate of completion transitions is $\mu$.
    Nonetheless, the long-run arrival and completion rate must be equal, because the system is stable.
    Thus the long-run completion rate is $\lambda$.
    The only way that a completion transition can not cause a completion is if the queue is empty,
    causing unused service.
    As a result, the rate of unused service must be $\mu - \lambda$.
\end{proof}

Note that by \cref{lem:unused}, in systems with exponential service such as the two-level arrivals system, $P(Q > 0) = \lambda/\mu = \rho$.

\section{Relative arrivals and completions}
\label{app:rel-arr-comp}

First, we verify that $\Delta_A$ and $\Delta_C$, the relative arrivals and relative completions functions,
are well-defined and finite. Note that $\Delta_A$ is the relative value function of a average-reward Markov Reward process
with Markov chain matching the arrival process and reward equal to the arrival rate $\lambda_i$. By \citet[Section~8.2.1]{puterman_markov_1994}, this relative value is well-defined and finite for any finite-state Markov chain. $\Delta_C$ is equivalent.

Next, we establish two systems of equations, equivalent to the Poisson equation for MDPs,
which can be used to calculate $\Delta_A(i)$ in each state $i\in\arrSet$ (or $\Delta_C(i)$ in each state $i\in \compSet$), up to an additive constant. 
Incorporating \cref{lem:mean-delta}, which states that $\Ep[\Delta_A(\Yarr)] = 0$ (or $\Ep[\Delta_C(\Yarr)] = 0$), uniquely characterizes $\Delta_A(i)$ (or $\Delta_C(i)$). 
\begin{lemma}
    \label{lem:poisson-delta}
    For any arrivals chain, the relative arrivals function $\Delta_A$ and relative completions function $\Delta_C$ satisfy the following systems of equations:
    \begin{align}
    \label{eq:poisson-delta-A}
    \forall i \in \arrSet, \quad \Delta_A(i) &= \frac{\lambda_{i}-\lambda}{r_{i,\cdot,\cdot}} + \sum_{j\in\arrSet} \frac{r_{i,j,\cdot}}{r_{i,\cdot,\cdot}} \Delta_A(j), \quad
    \forall i \in \compSet, \quad \Delta_C(i) = \frac{\mu_{i}-\mu}{s_{i,\cdot,\cdot}} + \sum_{j\in\compSet} \frac{s_{i,j,\cdot}}{s_{i,\cdot,\cdot}} \Delta_C(j).
    \end{align}
\end{lemma}
\begin{proof}
    Let us start with \cref{def:rel-arr-comp}: $\Delta_A(i) := \lim_{t\to\infty} \Ep[A_{i}(t)] - \lambda t$,
where $A_i(t)$ denotes the number of arrivals by time $t$,
given that the arrivals chain starts from state $i$.

Let us split up $\Ep[A_{i}(t)] - \lambda t$, the expression inside the limit,
into two time periods:
Until the arrival chain next undergoes a transition,
and after that time.
The transition occurs after $Exp(r_{i,\cdot,\cdot})$ time,
or $\frac{1}{r_{i,\cdot,\cdot}}$ expected time.
Prior to the transition,  $\Ep[A_i(t)] - \lambda t$ accrues at rate $\lambda_{i} - \lambda$.
We use $\Yarr^{next}$ to denote the state after the transition.
Then we have the expression: 
\begin{align*}
    \Delta_A(i) &= \frac{\lambda_{i} - \lambda}{r_{i,\cdot,\cdot}} + \lim_{t\to\infty} \Ep[A_{\Yarr^{next}}(t)] - \lambda t 
    =  \frac{\lambda_{i} - \lambda}{r_{i,\cdot,\cdot}} + \Ep[\Delta_A(\Yarr^{next})]. 
\end{align*}
Finally, note that $P(\Yarr^{next} = j) = \frac{r_{i,j,\cdot}}{r_{i,\cdot,\cdot}}$,
allowing us to derive the equations for $\Delta_A(i)$ in \eqref{eq:poisson-delta-A}.
The equations for $\Delta_C(i)$ can be proved similarly. 
\end{proof}

As a corollary of \cref{lem:poisson-delta}, 
we characterize the change in the value of $\Delta_A(i)$ when a transition occurs in the arrival chain,
and the change in the value of $\Delta_C(i)$ when a transition occurs in the completion chain. 

\begin{corollary}\label{cor:effect-arr-comp-on-rel}
In the MAMS system, the effect of arrival (or completion) on relative arrival (or relative completion) satisfies: 
\begin{align}
    \label{eq:effect-arr-on-rel-arr}
    \forall i \in \arrSet, \quad &\sum_{j\in\arrSet, a\in\{0,1\}} r_{i, j, a} (\Delta_A(j) - \Delta_A(i) + a) = \lambda, \\
    \label{eq:effect-comp-on-rel-comp}
    \forall i \in \compSet, \quad &\sum_{j \in \compSet, c\in\{0,1\}} s_{i, j, c} (\Delta_C(j) - \Delta_C(i) + c) = \mu. 
\end{align}
Equivalently, for any $i\in \arrSet$, $G\circ \Delta_A(i) = \lambda - \lambda_i$; for any $i \in \compSet$, $G\circ \Delta_C(i) =  \mu - \mu_i$. 
\end{corollary}

\begin{proof}
    By \eqref{eq:poisson-delta-A} in \Cref{lem:poisson-delta}, for any $i\in \arrSet$, $\Delta_A(i) = \frac{\lambda_{i}-\lambda}{r_{i,\cdot,\cdot}} + \sum_{j\in\arrSet} \frac{r_{i,j,\cdot}}{r_{i,\cdot,\cdot}} \Delta_A(j)$. 
    Multiplying both sides by $r_{i,\cdot,\cdot}$, rearranging the terms, and expanding $r_{i, \cdot, \cdot}, r_{i,j, \cdot},$ and $\lambda_i$, we get \eqref{eq:effect-arr-on-rel-arr}.
    \eqref{eq:effect-comp-on-rel-comp} follows similarly.

    It can be convenient to express \eqref{eq:effect-arr-on-rel-arr} and \eqref{eq:effect-comp-on-rel-comp}in terms of $G \circ \Delta_A$ and $G \circ \Delta_C$.
    Note that by \cref{lem:g-f-formula-mams},
        $G \circ \Delta_A(i) = \sum_{j\in\arrSet} r_{i, j, \cdot} (\Delta_A(j) - \Delta_A(i)).$
    Thus, $G\circ \Delta_A(i) = \lambda - \lambda_i$ and $G\circ \Delta_C(i) = \mu - \mu_i$.
\end{proof}

\section{Generator and drift method}\label{app:lemma-generator-drift}
\begin{replemma}{lem:drift-lemma-mams}
    Let $f$ be a test function for which $\Ep[f(Q, \Yarr, \Ycomp)] < \infty.$
    Then 
    \begin{align*}
        \Ep[G \circ f(Q, \Yarr, \Ycomp)] = 0.
    \end{align*}
\end{replemma}
\begin{proof}
    We apply Proposition 3 of \cite{glynn_bounding_2008}.
    To prove \cref{lem:drift-lemma-mams} using that proposition,
    we only need to verify that the total transition rate of the system is uniformly bounded.

    For each state $(q, i_A, i_C)$, the total transition rate of the MAMS system is the total transition rate of the arrival process and the completion process, i.e., $\sum_{j_A\in\arrSet, a} r_{i_A, j_A, a} + \sum_{j_C\in\compSet, c} s_{i_C, j_C, c}$, which is finite for any $i_A\in\arrSet$ and $i_C\in\compSet$. Because $\arrSet$ and $\compSet$ are both finite sets, the total transition rates of the MAMS system are uniformly bounded.
\end{proof}

Next, we will give sufficient condition for proving that $\Ep[f(Q, \Yarr, \Ycomp)] < \infty,$ allowing us to apply \cref{lem:drift-lemma-mams} to the test functions we use in this paper.
%
%
To prove our sufficient condition, we need \cite[Theorem~2.3]{hajek_hitting_1982}, restated below using the state representation of the MAMS system. 
\begin{lemma}
\label{lem:lyapunov-moment-bound}
    Consider a Markov chain with uniformly bounded total transition rates, and a Lyapunov function $V$ that satisfies the conditions below for any state $(q, i_A, i_C)$: $V(q, i_A, i_C) \geq 0$; there exists $b, \gamma > 0$ such that whenever $V(q, i_A, i_C) \geq b$, 
    \begin{equation}\label{eq:check-negative-drift}
        G\circ V(q, i_A, i_C) \leq -\gamma;
    \end{equation}
    there exists $d > 0$ such that for any initial state, $(q, i_A, i_C)$, and any state after one step of transition, $(q', j_A, j_C)$, 
    \begin{equation}\label{eq:check-bounded-jump}
        V(q',j_A, j_C) - V(q,i_A,i_C) \leq d. 
    \end{equation}
    Then there exists $\theta > 0$ such that $\Ep[e^{\theta V(Q, \Yarr, \Ycomp)}] < \infty.$
\end{lemma}

We also need a characterization of the drift of a generic test function in the MAMS system, \Cref{lem:g-f-formula-mams}.
\begin{replemma}{lem:g-f-formula-mams}
    For any real-valued function $f$ of the state of the MAMS system, 
    \begin{align*}
         G \circ f(q, i_A, i_C) &= \sum_{j_A\in\arrSet, a\in[0,1]} r_{i_A, j_A, a} \left(f(q+a, j_A, i_C) -  f(q, i_A, i_C)\right) \\
         &+ \sum_{j_C\in\compSet, c \in[0,1]} s_{i_C, j_C, c } \left(f(q-c+u, i_A, j_C) - f(q, i_A, i_C)\right),
    \end{align*}
    where $u=(c-q)^+$ denotes the unused service. 
\end{replemma}

\begin{proof}
    Because the generator is the stochastic equivalent of the derivative operator, to get an expression for $G \circ f(q, i_A, i_C)$, we examine all types of transitions in the MAMS system.  
    
    For each $j_A\in\arrSet$ and $a\in \{0,1\}$, with rate $r_{i_A, j_A, a}$, the state of the arrival process changes to $j_A$ accompanied by $a$ arrivals. This type of transition changes the state to $(q+a, j_A, i_C)$, so it contributes $r_{i_A, j_A, a} \left(f(q+a, j_A, i_C) - f(q, i_A, i_C)\right)$
    to $G \circ f(q, i_A, i_C)$. 
    
    Similarly, for each $i_C \in \compSet$ and $c\in \{0,1\}$, with rate $s_{i_C, j_C, c}$, the state of the completion process changes to $j_C$ accompanied by $c$ completions. This transition changes the state to $((q-c)^+, i_A, j_C)$, so it contributes $s_{i_C, j_C, c} \left(f((q-a)^+, i_A, j_C) - f(q, i_A, i_C)\right)$
    to $G\circ f(q, i_A, i_C)$. Let $u = (c-q)^+$, then $(q-a)^+ = q - a + u$. 

    We sum up all of the above terms to obtain the formula for $G\circ f(q, i_A, i_C)$.
\end{proof}

We are now ready to state and prove our sufficient condition for $\Ep[f(Q, \Yarr, \Ycomp)] < \infty$ using \Cref{lem:lyapunov-moment-bound} and \cref{lem:g-f-formula-mams}. \cref{lem:polynomial-finite-expectation} is analogous to Lemma~A.2 of \cite{grosof_marc_reset_2023}.  

 \begin{lemma}\label{lem:polynomial-finite-expectation}
    Let $f$ be a real-valued function of the state of the MAMS system, $(q, i_A, i_C)$. 
    Suppose $f(q, i_A, i_C)$ grows at a polynomial rate in $q$.
    Suppose that $\lambda < \mu$ in the MAMS system.
    Then $f$ has finite expectation in steady state:
    $\Ep[f(Q, \Yarr, \Ycomp)] < \infty$.
\end{lemma}
\begin{proof}
    We will verify conditions of \Cref{lem:lyapunov-moment-bound}. 
    Define $V(q, i_A, i_C) = (q + \Delta_A(i_A) - \Delta_C(i_C))^+.$
    Let $\Delta_{\max} = \max\{\max_{i_A\in\arrSet} \abs{\Delta_A(i_A)}, \max_{i_C\in\compSet} \abs{\Delta_C(i_C)}\}$. Letting $b = 1 + 4\Delta_{\max}$ and $\gamma = \mu - \lambda$. Consider any state $(q,i_A, i_C)$ such that $V(q,i_A, i_C) \geq b$.
    \begin{itemize}
        \item By the definition of $V(q,i_A, i_C)$, we have $q \geq 1 + 2\Delta_{\max}$, so $V(q, i_A, i_C) = q + \Delta_A(i_A) - \Delta_C(i_C)$. 
        \item Moreover, for any state $(q', j_A, j_C)$ reachable after one step of transition from $(q,i_A, i_C)$, we must have $V(q', j_A, j_C) \geq q - 1 + \Delta_A(j_A) - \Delta_C(j_C) \geq 0$, which implies that $V(q', j_A, j_C) = q' + \Delta_A(j_A) - \Delta_C(j_C)$. 
    \end{itemize}
    By \Cref{lem:g-f-formula-mams}, 
    \begin{align*}
        G\circ V(q, i_A, i_C) &= \sum_{j_A\in \arrSet, a} r_{i_A, j_A, a} (a +\Delta_A(j_A) - \Delta_A(i_A)) + \sum_{j_C\in\compSet, c} s_{i_C, j_C, c} (-c + u - \Delta_C(j_C) + \Delta_C(i_C)). 
    \end{align*}
    Because $q \geq 1 + 2\Delta_{\max}$ and $c \leq 1$, we have $u= (c-q)^+ = 0$. Therefore, \Cref{cor:effect-arr-comp-on-rel} implies that $G\circ V(q, i_A, i_C) = \lambda - \mu < 0.$
    This proves the condition \eqref{eq:check-negative-drift}. It is not hard to see that the other condition, \eqref{eq:check-bounded-jump}, holds with $d = 1 + 4\Delta_{\max}$. Then Theorem 2.3 of \cite{hajek_hitting_1982} implies the existence of $\theta > 0$ such that $\Ep\left[e^{\theta V(q, i_A, i_C)}\right] < \infty.$
    Observe that $e^{\theta V(q, i_A, i_C)}$ grows exponentially fast in $q$. Therefore, for any $f$ that grows at a polynomial rate in $q$,
    \begin{align*}
        f(q,i_A, i_C) &= O(e^{\theta V(q, i_A, i_C)}), \quad \Ep[f(Q, \Yarr, \Ycomp)] < \infty.\qedhere
    \end{align*}
\end{proof}

\end{document}